%% file: main.tex
\documentclass[runningheads]{llncs}
\usepackage[T1]{fontenc}
\usepackage{hyperref}
\input{macros}

\title{XALP-completeness of Parameterized Problems on Planar Graphs}
\author{Hans L. Bodlaender\inst{1}\orcidID{0000-0002-9297-3330} \and \\
	Krisztina Szil\'agyi\inst{2}\orcidID{0000-0003-3570-0528}\thanks{Supported by the project CRACKNP that has received funding from the European
		Research Council (ERC) under the European Union's Horizon 2020 research and
		innovation programme (grant agreement No 853234) and under the project Robotics and advanced industrial production (reg. no. CZ.02.01.01/00/22\_008/0004590).}}
\institute{Utrecht University, Utrecht, the Netherlands \email{h.l.bodlaender@uu.nl}\and Czech Technical University, Prague, Czechia \email{krisztina.szilagyi@fit.cvut.cz}}
\begin{document}
	
	\maketitle
	
	\begin{abstract}
		The class XNLP consists of (parameterized) problems that can be solved non-deterministically in $f(k)n^{O(1)}$ time and $g(k)\log n$ space, where $n$ is the size of the input instance and $k$ the parameter. 
		The class XALP consists of problems that can be solved in the above time and space with access to an additional stack. These two classes are a ``natural home'' for many standard graph problems and their generalizations. 
		
		In this paper, we show the hardness of several problems on planar graphs, parameterized by outerplanarity, treewidth and pathwidth, thus strengthening several existing results. 
		In particular, we show XALP-completeness of the following problems parameterized by outerplanarity: \textsc{All-or-Nothing Flow}, \textsc{Target Outdegree Orientation}, \textsc{Capacitated (Red-Blue) Dominating Set}, \textsc{Target Set Selection} etc. We also show the XNLP-completeness of \textsc{Scattered Set} parameterized by pathwidth and XALP-completeness parameterized by treewidth and outerplanarity.
	\end{abstract}
\section{Introduction}
	\label{sec:intro}
	In classical complexity theory, we can classify problems depending on their space complexity or their time complexity. 
	These two classifications are intertwined as follows: $L\subseteq NL\subseteq P\subseteq NP\subseteq PSPACE=NPSPACE\subseteq EXP$. In particular, the class NL (Nondeterministic Logarithmic-space) is contained in the class P (Polynomial-time).
	It is natural to ask whether there is a similar chain of inclusions for the parameterized analogues of these classes. In particular, what can we say about the relationship between FPT (analogue of P) and XNL (analogue of NL)? 
	
	This question was first posed by Elberfeld et al.~\cite{ElberfeldST15}. They defined the class $N[fpoly, flog]$ (here called XNLP) as the class of (parameterized) problems that can be solved with a nondeterministic Turing machine 
	in $f(k)n^{O(1)}$ time and $g(k)\log n$ space and gave the first problems complete for this class. They also present several problems that are complete for this class, stating: \textit{The real challenge lies in finding problems together with \emph{natural} parameterizations that are complete}. 
	
	In the last decade, this question has been resolved: many standard graph problems and their generalizations were shown to be XNLP-complete (see Table~\ref{table:results} for examples of such problems). This class turned out to be a natural home for several problems that are $W[t]$-hard for all $t$ but unlikely to belong to $W[P]$, such as \textsc{Bandwidth} and \textsc{List Coloring} (parameterized by pathwidth) \cite{BodlaenderGNS22a}.  
	
	Recently, in~\cite{BodlaenderGJPP22a} a ``tree variant'' of the class XNLP was introduced. This class is called XALP, and many problems that are XNLP-complete parameterized by pathwidth are shown to be XALP-complete parameterized by treewidth (see Table~\ref{table:results}). The class XALP contains problems that can be solved by a nondeterministic Turing machine with access to an auxiliary stack in $f(k)n^{O(1)}$ time and $g(k)\log n$ space.
	
	A natural question to ask is whether we can say something about hardness for other graph parameters. Several XNLP-hard problems for treewidth are known to be easy (i.e. in FPT) parameterized by the so-called stable gonality parameter~\cite{BodlaenderCW22a}. In this paper, we study the hardness of problems on planar graphs, parameterized by outerplanarity, treewidth and pathwidth.
	
	Outerplanarity is a natural parameter for studying problems on planar graphs. As graphs of outerplanarity
	$k$ have treewidth at most $3k-1$ (see, e.g.,~\cite{bodlaender1998partial}), algorithmic results for graphs 
	of bounded treewidth carry over to planar graphs of bounded outerplanarity. This has also been exploited
	in the well-known and often applied layering technique of Baker~\cite{Baker94}, resulting in approximation schemes for many problems on planar graphs.
	
	There are several celebrated results in the field of parameterized
	complexity where the complexity of a problem significantly decreases
	when we move from general graphs to planar graphs, for instance
	\textsc{Dominating Set} parameterized by solution size
	then drops from being W[2]-complete~\cite{downey1995fixed} to a fixed parameter tractable problem with a linear kernel~\cite{AlberFN04}.
	The results of this paper show that for several problems that
	are in XP for treewidth as parameter, we have no complexity drop
	when we move to the realm of planar graphs.
	
	\paragraph{Our results.} An overview of the results of this paper can be seen in Table \ref{table:results}. 
	In addition to these results, we also show the XNLP-completeness of \textsc{Binary CSP} for $k\times n$-grids (parameterized by $k$). 
    
	\renewcommand{\arraystretch}{1.5}
	\setlength{\tabcolsep}{3pt}
	\begin{table}[h!]
		\centering
		\begin{tabular} {|L{10em}|C{8em}|C{8.5em}|C{8.5em}|} \hline
			& outerplanarity & treewidth & pathwidth \\ \hline
			\textsc{Binary CSP} & XALP-c. (\ref{ssec:csp_tw}) & XALP-c. \cite{BodlaenderGJPP22a} & XNLP-c. \cite{BodlaenderGNS22a}\\ \hline
			\textsc{List Coloring} & XALP-c. (\ref{ssec:csp_tw}) & XALP-c. \cite{BodlaenderGJPP22a} & XNLP-c. \cite{BodlaenderGNS22a}\\ \hline
			\textsc{Precoloring Extension} & XALP-c. (\ref{ssec:csp_tw})  & XALP-c. \cite{BodlaenderGJPP22a} & XNLP-c. \cite{BodlaenderGNS22a}\\ \hline
			\textsc{Scattered Set} & XALP-c. (\ref{sec:scattered}) & XALP-c. (\ref{sec:scattered}) & XNLP-c. (\ref{sec:scattered})\\ \hline
			\textsc{All-or-Nothing Flow} & XALP-c. (\ref{sec:flow}) & XALP-c. \cite{BodlaenderGJPP22a} & XNLP-c. \cite{BodlaenderCW22a}\\ \hline
			\textsc{All-or-Nothing Flow} with cap. $\leq 2$ & XALP-c. (\ref{ssec:flow12}) & XALP-c. (\ref{ssec:flow12})  & XNLP-c. (\ref{ssec:flow12}) \\ \hline
			\textsc{Target Outdegree Orientation} & XALP-c. (\ref{ssec:too}) & XALP-c. \cite{BodlaenderGJPP22a} & XNLP-c. \cite{BodlaenderCW22a} \\ \hline
                \textsc{Min./Max. Outdegree Orientation} & XALP-c. (\ref{ssec:too}) & XALP-c. \cite{BodlaenderGJPP22a} & XNLP-c. \cite{BodlaenderCW22a} \\ \hline
                \textsc{Circulating Orientation} & XALP-c. (\ref{ssec:co}) & XALP-c. \cite{BodlaenderGJPP22a} & XNLP-c. \cite{BodlaenderCW22a}\\ \hline
			\textsc{Capacitated (Red-Blue) Dominating Set} & XALP-c. (\ref{ssec:capds}) & XALP-c. \cite{BodlaenderGJPP22a} & XNLP-c. \cite{BodlaenderGJJL22}\\ \hline
			\textsc{Capacitated Vertex Cover} & XALP-c. (\ref{ssec:capvc}) & XALP-c. \cite{BodlaenderGJPP22a} & XNLP-c. \cite{BodlaenderGJJL22} \\ \hline
			$f$-\textsc{Dominating Set} & XALP-c. (\ref{ssec:fkdomination}) & XALP-c. (\ref{ssec:fkdomination}) & XNLP-c. (\ref{ssec:fkdomination}) \\ \hline
			$k$-\textsc{Dominating Set} & XALP-c. (\ref{ssec:fkdomination}) & XALP-c. (\ref{ssec:fkdomination}) & XNLP-c. (\ref{ssec:fkdomination})\\ \hline
			\textsc{Target Set Selection} & XALP-c. (\ref{ssec:tss}) & XALP-c. (\ref{ssec:tss}) & XNLP-c. (\ref{ssec:tss})\\ \hline
		\end{tabular}
		\vspace{5pt}
		\caption{An overview of results of this paper and results from the literature. Results of this paper are stated with the corresponding section numbers.}
		\label{table:results}    
	\end{table}
	
	\paragraph{Organization.} In Section~\ref{sec:notation} we introduce some definitions and notations. In Section~\ref{sec:csp} we strengthen two previous hardness results concerning \textsc{Binary CSP}. Section~\ref{sec:scattered} discusses the \textsc{Scattered Set} problem.
	In Section~\ref{sec:flow} we show hardness of the \textsc{All-or-Nothing Flow} problem. 
	In Section~\ref{sec:reductionsfromallornothing} we consider several problems whose hardness can be shown by a reduction from \textsc{All-of-Nothing Flow}.
	We finish with some concluding remarks in Section~\ref{sec:conclusion}. 
	
\section{Definitions and Notation}
	\label{sec:notation}
	\subsection{Graph notions}
	Given an integer $n$, we define $[n]=\{1,2,\dots,n\}$.
	Throughout this paper, graphs are undirected, unless otherwise specified.
	Given a graph $G=(V, E)$, we denote by $N(v)$ the open neighbourhood of a vertex $v$, i.e. $N(v)=\{u\in V:\: \{u,v\}\in E\}$.
	
	In Section~\ref{sec:flow}, we consider the \textsc{All-or-Nothing Flow} problem, a variant of the well-known \textsc{Flow} problem.
	
	\begin{definition}
		A \emph{flow network} is a tuple $(G,s,t,c)$, with $G=(V,E)$ a directed graph, $s,t\in V$ two
		vertices, and $c: E \rightarrow \mathbb{N}$ a capacity function, assigning to each arc
		a positive integer capacity. A \emph{flow} is a function $f: E \rightarrow \mathbb{N}$, that
		assigns to each arc a non-negative integer, such that
		\begin{enumerate}
			\item for each arc $e\in E$, $0 \leq f(e)\leq c(e)$ (capacity constraint), and
			\item for each vertex $v\in V\setminus \{s,t\}$, $\sum_{wv\in E} f(wv) = \sum_{vx\in E} f(vx)$ (flow conservation).
		\end{enumerate}
		The \emph{value} of a flow
		$f$ is $\sum_{sx\in E} f(sx) - \sum_{ws\in E} f(ws)$.
		A flow $f$ is an
		\emph{all-or-nothing} flow, if for each edge $e\in E$, $f(e)\in \{0,c(e)\}$.
	\end{definition}

Note that $s$ (respectively $t$) might have incoming (outgoing) edges with nonzero flow. For this reason, we define the value of a flow as the difference between the outgoing and incoming flow to $s$ as opposed to just the outgoing flow from $s$.  
    
	One of the parameters we consider is outerplanarity.
	A planar graph is outerplanar (or 1-outerplanar) if it has an embedding such that all its vertices are on the outer face. Informally speaking, a planar graph $G$ is $k$-outerplanar if its vertices are on $k$ ``layers''. Formally, $k$-outerplanar graphs are defined as follows.
	
	\begin{definition}
		An embedding of a planar graph $G$ is $k$-\emph{outerplanar} if for every vertex $v$ there is a sequence $F_1,v_1,F_2,v_2,\dots, F_\ell, v_\ell$ such that:
		\begin{itemize}
			\item $F_1$ is the unbounded face and $v_\ell=v$;
			\item For each $i\in [\ell]$, $F_i$ contains $v_i$;
			\item For each $i\in [\ell-1]$, $v_i$ belongs to $F_{i+1}$;
			\item $\ell\leq k$.
		\end{itemize}  
		A planar graph is $k$-\emph{outerplanar} if it has a $k$-outerplanar embedding.
	\end{definition}
	
	In other words, a graph is $k$-outerplanar if after $k$ operations consisting of removing all the vertices on the outer face, we obtain an empty graph. The outerplanarity of a graph
	can be computed in polynomial time~\cite{Kammer07}.
	
	\paragraph*{Treewidth and Pathwidth.} For completeness, we state the standard definition of treewidth and pathwidth, as well as another equivalent definition that we will use in the later sections.
	
		\begin{definition}
		A \emph{tree decomposition} of a graph $G=(V,E)$ is a pair $(\{X_i:\: i\in I\}, T=(I,F))$, with
		$T$ a tree and $\{X_i~|~i\in I\}$ a collection of subsets of $V$, such that
		\begin{enumerate}
			\item $\bigcup_{i\in I} X_i = V$;
			\item For each edge $\{v,w\}\in E$, there is an $i\in I$ with $v, w\in X_i$; and
			\item For all $v\in V$, the set $I_v =\{i\in I:\:  v\in X_i\}$ forms a connected subtree of $T$.
		\end{enumerate}
		The \emph{width} of a tree decomposition $(\{X_i:\: i\in I\}, T=(I,F))$ is $\max_{i\in I} |X_i|-1$, and the \emph{treewidth} of a graph $G$ is the minimum width of a tree decomposition 
		of $G$.
		
		A tree decomposition $(\{X_i:\: i\in I\}, T=(I,F))$ is a \emph{path decomposition}, if $T$ is a 
		path, and the \emph{pathwidth} of a graph $G$ is the minimum width of a path decomposition of $G$.
	\end{definition}
	
	Throughout the paper,  we assume that we are given the tree decomposition and path decomposition of the input graph.
	Kloks~\cite{Kloks94} introduced the notion of nice tree and path decompositions. We will use a 
	variant of nice path decompositions, which we will describe with a sequence of operations
	on terminal graphs. Here, a terminal graph is a triple $(V,E,X)$, together with a binary relation $\prec\subset X\times X$ which is  astrict total ordering on $X$. We call elements of $X$ terminals. 
	We consider the following operations on terminal graphs: 
	\begin{itemize}
		\item \textbf{Introduce}. Given a terminal graph $G=(V,E,X)$, the introduce operation adds
		a new isolated vertex $v$ to $V$ and to $X$, with $v$ the smallest vertex regarding the ordering in $X$: $G'=(V\cup \{v\}, E, \{v\} \cup X)$, with $v \prec x$ for all $x\in X$.
		\item \textbf{Forget}. Given a terminal graph $G=(V,E,X)$ with $X\neq\emptyset$, the largest
		element of $X$ is no longer a terminal: $G'= (V, E, X \setminus \{x\})$, with
		$y \prec x$ for all $y\in X \setminus \{x\}$.
		\item \textbf{Add-Edge}$(i)$. Given a terminal graph $G=(V,E,X)$ with $|X|>i$, we add an edge
		between the $i$th and the $(i+1)$st terminal: $V$ and $X$ are unchanged, and one edge is added to $E$.
		\item \textbf{Swap}$(i)$. Given a terminal graph $G=(V,E,X)$ with $|X|>i$, swap in the ordering
		in $X$ the $i$th and the $(i+1)$st vertex. 
	\end{itemize}
	
	To build graphs of small treewidth, we need one more operation, that now has
	two terminal graphs as input.
	\begin{itemize}
		\item \textbf{Join}. Given two terminal graphs $G_1=(V_1, E_1, X)$ and
		$G_2=(V_2,E_2,X)$ that only intersect in their terminal nodes, i.e., $V_1\cap V_2 = X$ and $E_1\cap E_2=\emptyset$, we build the graph $G=(V_1\cup V_2, E_1\cup E_2, X)$. The operation `fuses' the two terminal graphs, identifying 
		for each $i$, the $i$th terminal node of the first graph with the $i$th terminal node of the second graph.
	\end{itemize}
	
	Note that with a sequence of Swap operations, we can order the terminals in any manner; one can
	thus easily observe that a nice path decomposition of a graph $G =(V,E)$ of pathwidth at most $k$ (see~\cite{Kloks94}) can be transformed
	into a sequence of Introduce, Forget, Add-Edge and Swap operations that creates the terminal
	graph $G=(V,E,\emptyset)$ with $O(kn)$ operations, 
	and with each intermediate terminal graph
	having at most $k+1$ terminals. 
    
    Similarly, a nice tree decomposition
	can be transformed to a similar sequence of  Introduce, Forget, Add-Edge, Swap
	and Join operations. From this sequence of operations, we can go back to
	a \emph{generalized tree decomposition}, i.e., we have a rooted tree, with each bag
	consisting of the set of terminals, and of one of the following types: Introduce, Forget, Add-Edge, Swap
	and Join.
	We omit the simple details; one can also carry out these
	transformations in logarithmic space. 
	
	\subsection{The classes XNLP and XALP}
	A \emph{parameterized problem} is a language $L\subseteq \Sigma^*\times\mathbb{N}$ for some finite alphabet $\Sigma$. 
	The class XNLP consists of parameterized problems that can be solved by a non-deterministic algorithm which uses $f(k)n^{O(1)}$ time and $g(k)\log n$ space for some computable function $f$ (where $k$ is the parameter and $n$ is the size of the input instance).
	
	As in the classical complexity setting, in order to define XNLP-hardness and XNLP-completeness, we need the notion of \emph{reduction}. 
	\begin{definition}
		A \emph{parameterized reduction} from a problem $L_1\subseteq \Sigma_1^*\times \mathbb{N}$ to a problem $L_2\subseteq \Sigma_2^*\times \mathbb{N}$ is a function $f:\Sigma_1^*\times \mathbb{N}\rightarrow \Sigma_2^*\times \mathbb{N}$ such that the following two conditions are satisfied:
		\begin{itemize}
			\item For all $(x,k)\in \Sigma_1^*\times \mathbb{N}$, $(x,k)\in L_1$ if and only if $f((x,k))\in L_2$,
			\item There is a computable function $g:\mathbb{N}\rightarrow \mathbb{N}$ such that for all $(x,k)\in \Sigma_1^*\times \mathbb{N}$ we have $k'\leq g(k)$, where $f((x,k))=(y,k')$.
		\end{itemize}
		We call $f$ a \emph{parameterized logspace reduction} or \emph{pl-reduction} if there is an algorithm that computes $f((x,k))$ in space $O(g'(k)+\log |x|)$, where $g':\mathbb{N}\rightarrow \mathbb{N}$ is a computable function and $|x|$ the size of $x$ (i.e. number of bits needed to store $x$). 
	\end{definition}
	
	We call a problem XNLP-hard if any other problem in XNLP can be pl-reduced to it, and XNLP-complete if it is XNLP-hard and in XNLP.
	
	We define XALP as the class of problems that can be solved by a non-deterministic Turing machine with a stack in $f(k)n^{O(1)}$ time and $g(k)\log n$ space. For other equivalent definitions of XALP we refer the reader to~\cite{BodlaenderGJPP22a}. 
	
	Membership in XNLP (respectively XALP) can usually be shown by standard dynamic programming techniques on path decompositions (respectively tree decompositions). Intuitively, the $\log$ space complexity can be achieved by guessing the entry (rather than storing all the entries). Examples of membership proofs can be found in \cite{BodlaenderGJPP22a,BodlaenderGJJL22,BodlaenderGNS22a}. In particular, note that bounded outerplanarity implies bounded treewidth, so any problem which is in XALP (XNLP respectively) parameterized by treewidth is in XALP (XNLP) parameterized by outerplanarity. 
	
	In our hardness proofs, we skip the details that show the logarithmic space bounds for the
	pl-reductions. The standard technique to realize such space is to recompute needed values instead of storing them. For example, if the reduction uses a Sidon set, we do not store the
	Sidon set, but each time we need the $i$th element of the set, we recompute it. This increases
	the time of the reduction by a polynomial factor, but keeps the used space small.
	
\section{Binary CSP}
	\label{sec:csp}
	The \textsc{Constraint Satisfaction} Problems (CSP) form a well-known class of problems with a wide set of applications, ranging from artificial intelligence to operations research. 
	In this section, we study \textsc{Binary CSP}, a CSP where all constraints are binary. Formally, the problem is defined as follows.
	
	\defproblem{\textsc{Binary CSP}}{A graph $G = (V,E)$, a set $\mathcal{C}$, set $C(v)\subseteq \mathcal{C}$ for each $v\in V$, set $C(u,v)\subseteq \mathcal{C}\times\mathcal{C}$ for each ordered pair $(u,v)\in V\times V$ such that $\{u,v\}\in E$}{Is there a function $f:V\rightarrow \mathcal{C}$ such that for every $v\in V$, $f(v)\in C(v)$ and for every $uv\in E$, we have $(f(u), f(v))\in C(u,v)$?}
	
	We call the elements of $\mathcal{C}$ colors, the sets $C(v)$ domains (or vertex constraints) and the sets $C(u,v)$ edge constraints. We emphasize that $C(u,v)$ contains ordered pairs, i.e. even though the graph is undirected, the ordering of vertices matters for the edge constraints. 
	
	\textsc{List Coloring} and \textsc{Precoloring Extension} are special cases of \textsc{Binary CSP}. 
	\textsc{List Coloring} and \textsc{Precoloring Extension} are special cases of \textsc{Binary CSP}, defined as follows.
	
	\defproblem{\textsc{List Coloring}}{A graph $G = (V,E)$, a set $\mathcal{C}$, set $C(v)\subseteq \mathcal{C}$ for each $v\in V$}{Is there a function $f:V\rightarrow \mathcal{C}$ such that for every $v\in V$, $f(v)\in C(v)$ and for every $uv\in E$, $f(u)\neq f(v)$?}
	
	\defproblem{\textsc{Precoloring Extension}}{A graph $G = (V,E)$, a set $\mathcal{C}$, a subset $W\subseteq V$, a function $f': W \rightarrow \mathcal{C}$}{Is there a function $f:V\rightarrow \mathcal{C}$ such that for every $v\in W$, $f(v)=f'(v)$ and for every $uv\in E$, $f(u)\neq f(v)$?}
	
	Note that an instance of \textsc{List Coloring} can be seen as an instance of \textsc{Binary CSP} where for each edge $uv$ the edge constraint is $C(u,v)=\{(a,b):a\neq b\}$. An instance of \textsc{Precoloring Extension} can be seen as an instance of \textsc{Binary CSP} where each vertex has a domain of size 1 or $|\mathcal{C}|$.
	Note that we may assume without loss of generality that $C(u,v)\subseteq C(u)\times C(v)$. 
	
	\textsc{Binary CSP}, \textsc{List Coloring} and \textsc{Precoloring Extension} were shown to be XNLP-complete parameterized by pathwidth \cite{BodlaenderGNS22a}. In the first part of this section, we will show that \textsc{Binary CSP} is also XNLP-complete for a more restrictive graph class, namely $k\times n$-grids (parameterized by $k$).
	As vertices whose lists are larger than their degree plus one can always be colored, and
	\textsc{List Coloring} is FPT for graphs of bounded treewidth plus maximum list size
	\cite{JansenS97}, \textsc{List Coloring} and \textsc{Precoloring Extension} are FPT for
	$k\times n$-grid, parameterized by $k$.
	
	\textsc{Binary CSP}, \textsc{List Coloring} and \textsc{Precoloring Extension} were shown to be XALP-complete parameterized by treewidth and by treewidth plus degree \cite{BodlaenderGJPP22a}. In the second part of this section, we will show for each of these three problems XALP-completeness for a more restrictive graph class, namely $k$-outerplanar graphs (parameterized by $k$). We remark that \textsc{List Coloring} is in L for trees~\cite{BodlaenderGJ22listcoloring}.
	
	\subsection{XNLP-completeness for $k\times n$-grids}
	\label{ssec:csp_grid}
	In order to prove the XNLP-completeness for $k\times n$-grids, we will use the following lemma, which gives an equivalent definition of pathwidth.
	\begin{lemma}[\cite{bodlaender1998partial}]\label{lem:pw_int}
		A graph $G$ has pathwidth $k$ if and only if it is a subgraph of an interval graph $H$ such that $\omega(H)=k+1$.
	\end{lemma}
	Now we are ready to prove the main theorem of this section.
	\begin{theorem}
		\textsc{Binary CSP} is XNLP-complete for $k\times n$-grids parameterized by $k$.  
	\end{theorem}
	\begin{proof}
		We will reduce from \textsc{Binary CSP} parameterized by pathwidth. Let $G$ be a graph of pathwidth $k-1$. By Lemma~\ref{lem:pw_int}, there is an interval graph $H'$ such that $G$ is a subgraph of $H'$ and $\omega(H')=k$. Let $\mathcal{I}$ be a collection of intervals whose intersection graph is $H'$. Intuitively, we will draw the intervals in $\mathcal{I}$ on a grid. Formally, we construct an instance of \textsc{Binary CSP} with underlying graph $H$ as follows. We start with $H$ being a $k\times n$-grid. To each $v\in V(G)$, we assign a horizontal path $I_v$ and a vertex $\ell_v$ in $H$ such that the following conditions are met:
		\begin{itemize}
			\item $\ell_v$ is the left endpoint of $I_v$;
			\item $I_v\cap I_u=\emptyset$ for all $u\neq v$;
			\item $\ell_u$ and $\ell_v$ are in different columns for all $u\neq v$;
			\item For all edges $uv\in E(G)$, there is a column in $H$ that contains a vertex in $I_v$ and a vertex in $I_u$.
		\end{itemize}
		We will now make the grid $H$ twice finer, i.e. we will add a row between every two adjacent rows and a column between every two adjacent columns. We will call the newly added rows, columns and vertices green, and the original vertices black. Each green vertex whose left and right neighbours are in $I_v$ for some $v\in V(G)$ is added to $I_v$.
		
		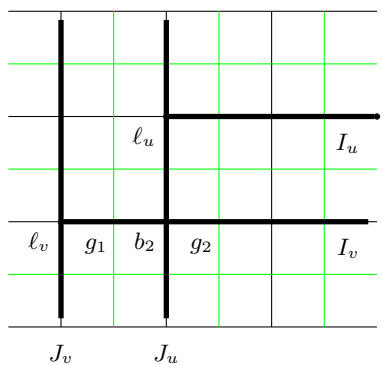
\begin{figure}[ht]
			\centering
			\import{figures/}{grid.tex}
			\caption{A part of the graph $H$ showing an intersection of intervals corresponding to $u$ and $v$}
			\label{fig:grid}
		\end{figure}
		
		For each vertex $v$, let $J_v$ be the column that contains $\ell_v$. Consider the green vertices on the paths $I_v$ and $J_v$ for some vertex $v\in V(G)$. We will add constraints to these vertices and to edges between adjacent vertices in $I_v$ and $J_v$ to ensure that all of them get the same color as $\ell_v$. Denote the vertices on $I_v$ from left to right by $b_1,g_1,b_2,g_2,\dots, b_{t-1},g_{t-1},b_t$, where $b_1=\ell_v$, the vertices $g_i$ are green, and $b_i$ are black (see Figure \ref{fig:grid}). 
		For $g_1$, we set $C(g_1)=C(v)$ and $C(\ell_v, g_1)=\{(c,c):\: c\in C(v)\}$. For $i=2,...,t$, we define the domains of $g_i$ and $b_i$, as well as the constraints of the adjacent edges as follows.
		We define $C(g_i)=C(v)$.
		If $b_i$ is in column $J_u$ for some $u\in V(G)$, we define $C(b_i)=\{c_{xy}:\: (x,y)\in C(v,u) \}$ (if $uv\not\in E(G)$, we define $C(u,v)=C(u)\times C(v)$). 
		We define $C(g_{i-1}, b_i)=\{(x, c_{xy}): xy\in C(v,u) \}$ and $C(b_i, g_i)=\{(c_{xy}, x): xy\in C(v,u)\}$. If $b_i$ is not in column $J_u$ for any $u\in V(G)$, we define $C(b_i)=\{c_{xx}:\: x\in C(v)\}$, $C(g_{i-1}, b_i)=\{(x, c_{xx}):\: x\in C(v)\}$ and $C(b_i, g_i)=\{(c_{xx}, x): \:x\in C(v)\}$.
		Note that now all of the green vertices on $I_v$ have the same color as $\ell_v$. We define the domains of vertices on $J_v$ and the constraints between adjacent vertices analogously. 
		For all the other vertices $v$ and edges $uv$ of $H$, we define $C(v)=\mathcal{C}$ and $C(u,v)=\mathcal{C}\times \mathcal{C}$.
		
		Given a valid coloring of $H$, we can color each vertex $v\in V(G)$ with the color of $\ell_v$.  
		For every edge $uv\in E(G)$, we have that $I_v$ and $J_u$ intersect at a black vertex or that $I_u$ and $J_v$ intersect at a black vertex. By our choice of domain for the black vertex in the intersection, we ensured that the constraint of the edge $uv$ is satisfied. Thus we get a valid coloring of $G$. It is easy to see that given a valid coloring of $G$, we can construct a valid coloring of $H$.
	\end{proof}

	\subsection{XALP-completeness parameterized by outerplanarity}
	\label{ssec:csp_tw}
	In this section, we show the XALP-completeness of \textsc{Binary CSP} parameterized by outerplanarity.
	\begin{theorem}
		\label{thm:csp_outerplanarity}
		\textsc{Binary CSP} is XALP-complete parameterized by outerplanarity.
	\end{theorem}
	\begin{proof}
		We will reduce from \textsc{Binary CSP} parameterized by treewidth. Given a graph $G$ with a nice tree decomposition of treewidth $k$, we will construct a $k$-outerplanar (plane) graph $H$. We first transform the nice tree decomposition to a generalized one, i.e., each bag is of one of the 
		following types:  Introduce, Forget, Add-Edge, Swap or Join.
		
		We denote by $\prec_X$ the ordering of terminals corresponding to the bag $X$. For each bag $X$ and each terminal $v$ in $X$, we add to $H$ a vertex $v_X$. We draw the vertices of $H$ in the plane as follows. For each bag $X$, the terminals in $X$ are on the same horizontal line, ordered from left to right according to $\prec_X$. For each bag $X$ and its child $Y$, the vertices in $H$ corresponding to $X$ are above vertices corresponding to $Y$.
		
		Denote by $B_X$ the set $B_X=\{v_{X}:\: v\in X\}\subseteq V(H)$.   
		For each node $X$ of the tree decomposition of $G$, we add the following edges to $H$. For each $v\in V(G)$, if both $X$ and its child $Y$ contain $v$, we add the edges $v_{X}v_{Y}$. In addition, if $X$ is an Add-Edge node corresponding to edge $vw$, we add the edge $v_Xw_X$ to $H$ (see Figure~\ref{fig:XALPtree}). 
		
		\begin{figure}[h]
			\centering
			\includegraphics{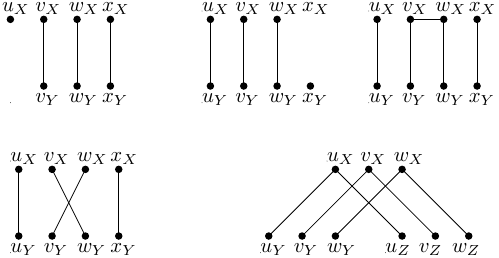}
			\caption{Representation of nodes: Introduce $u$, Forget $x$, Add-Edge $vw$, Swap $vw$, Join.}
			\label{fig:XALPtree}
		\end{figure}
		
		For each bag $B_X$ and each vertex $v_{X}\in B_X$, set $C(v_{X})=C(v)$. For each edge $v_{X}v_{Y}$ in $H$, set $C(v_{X}, v_{Y})=\{(d, d):\: d\in C(v)\}$. For each edge $v_Xw_X$ in $H$ set $C(v_{X}, w_{X})=C(v, w)$.
		
		Note that the graph $H$ is not planar yet. We will now turn it into a planar graph by using crossover gadgets. 
		Consider edges $ab$ and $cd$ that intersect at a point. Add a vertex $m$ at the intersection point and subdivide the edges $mb$ and $md$ by adding vertices $a'$ and $c'$ respectively. Set $C(a')=C(a)$, $C(c')=C(c)$, $C(m)=\{m_{ij}: \: i\in C(a), j\in C(c)\}$, $C(m, a)=C(m, a')=\{(m_{ij}, i): \: i\in C(a), j\in C(c)\}$, $C(m, c)=C(m, c')=\{(m_{ij}, j):\: j\in C(c), i\in C(a)\}$, $C(a', b)=C(a,b)$, $C(c', d)=C(c,d)$. This way, we ensure that $a'$ (respectively $c'$) has the same color as $a$ (respectively $c$).  
		
		In any valid coloring of $H$, all vertices that correspond to the same vertex in $G$ must have the same color. Also, for each edge in $G$, we have an Add-Edge node corresponding to it, and in that node we check whether the edge consraint is satisfied. Therefore, any valid coloring of $H$ defines a valid coloring of $G$ and vice versa. It remains to show that $H$ has bounded outerplanarity. 
		For each bag $X$, the vertices $v_{X}$ are on the outer face after at most $k/2$ rounds. The subgraphs corresponding to join nodes are subgraphs of a $k\times n$-grid (with some subdivided edges) and thus $k/2$-outerplanar. The vertices corresponding to the intersection of edges in the Swap nodes will be on the other face in at most $k/2+1$ rounds. Therefore, the graph $H$ is $(k/2+1)$-outerplanar. 
	\end{proof}
	
	The above theorem can be strenghtened: namely, it turns out that \textsc{Binary CSP} is XALP-complete on a subclass of planar graphs. This strenghtening will allow us to show XALP-completeness of \textsc{All-or-Nothing Flow} in Section~\ref{sec:flow}.
	
	\begin{corollary}
    \label{cor:csp_special}
		\textsc{Binary CSP} is XALP-complete parameterized by outerplanarity,
		for planar graphs that are connected, bipartite, have minimum degree two, maximum degree three, and adjacent vertices have disjoint colour sets.
	\end{corollary}
	
	\begin{proof}
		The reduction described in Theorem~\ref{thm:csp_outerplanarity} yields a connected graph of minimum degree two.
		A vertex $v$ of degree $d(v)$ larger than three can be split without changing the outerplanarity: replace $v$ by a tree with $d(v)$ leaves, and require for
		all edges inside the tree that the endpoints have the same colour.
		
		In order to make the graph bipartite, for each edge $\{v,w\}\in E$ we do the following. We 
		subdivide the edge by adding a vertex $x_{vw}$. The vertex $x_{vw}$ has colour set $C(v,w)$, i.e., the set of all allowed pairs of colours for $v$ and
		$w$. Now set $C(v, x_{vw}) = \{(c_v, (c_v,c_w)~|~ c_v\in C(v), 
		(c_v,c_w)\in C(v,w)\}$; i.e., we request that the first coordinate
		of the colour of $x_{c,v}$ is the colour of $v$. 
		Similarly, we set $C(w, x_{vw}) = \{(c_w, (c_v,c_w)~|~ c_w\in C(w), 
		(c_v,c_w)\in C(v,w)\}$. It is easy to see that this gives an equivalent, bipartite instance.
	\end{proof}
	
	It is easy to see that \textsc{Binary CSP} instances can be transformed into equivalent instances of \textsc{List Coloring} and \textsc{Precoloring Extension}, so Theorem~\ref{thm:csp_outerplanarity} implies hardness for these two problems as well.
	\begin{corollary}
		\textsc{List Coloring} and \textsc{Precoloring Extension} are XALP-complete
		parameterized by outerplanarity.
	\end{corollary}
	
	\begin{proof}
		We observe that the existing transformations increase the outerplanarity by at most 1, see
		e.g.~\cite{BodlaenderGNS22a}. An instance of \textsc{Binary CSP} can be transformed to
		an equivalent instance of \textsc{List Coloring} by renaming some colors, and
		replacing edges $uv$ by a number of vertices of degree 2, each adjacent to $u$ and to $v$.
		From \textsc{List Coloring} we go to \textsc{Precoloring Extension} by adding to each
		vertex $v$ precolored neighbors of degree 1; one such vertex for each color in $\mathcal{C}\setminus C(v)$.
	\end{proof} 
\section{Scattered Set}
	\label{sec:scattered}
	In this section, we will consider the \textsc{Scattered Set} problem,
	which is a generalization of the \textsc{Independent Set} problem and is defined as follows.
	
	\defproblem{\textsc{Scattered set}}{$G = (V,E)$, $k\in \mathbb{N}$, $d\in \mathbb{N}$}{Is there a set $I\subseteq V$ of size $k$ such that for any pair of distinct vertices $u,v\in I$ the distance between $u$ and $v$ is at least $d$?}
	
	For a fixed $d\geq 3$, the problem with parameter $k$ is W[1]-hard, even for bipartite graphs \cite{eto2014distance}. The case when $d=2$ gives
	the \textsc{Independent Set} problem which is W[1]-complete~\cite{DowneyF95fixed2}. For fixed $d$, the problem
	is FPT when parameterized by treewidth, by standard dynamic programming techniques (see \cite{katsikarelis2022structurally}). In
	this section, we consider the case that $d$ is part of
	the input.
	\begin{theorem}
		\textsc{Scattered Set} is XALP-complete parameterized by treewidth and by outerplanarity and it is XNLP-complete parameterized by pathwidth.
	\end{theorem} 
	\begin{proof}
		Membership in XNLP and XALP, respectively, follows with standard arguments, cf. the discussion in Section~\ref{sec:notation}. A state of the machine will contain as information
		the current bag in the decomposition, and for each vertex in this bag its distance to
		an element in the scattered set (i.e. $k+1$ integers in $[0,d]$).
		
		In both cases, we will reduce from \textsc{Binary CSP}. Given an instance of \textsc{Binary CSP} with underlying graph $G$ and $d$ colors, we will construct an instance $H$ of \textsc{Scattered Set} as follows. The distance between solution vertices $d$ equals
		the number of colors in $\mathcal{C}$.
		
		Firstly, for every vertex $v$ of $G$ construct a cycle $C_v$ of length equal to the degree of $v$. For every edge $\{u,v\}\in E(G)$, construct an edge in $H$ between a vertex in $C_u$ and $C_v$ such that the order of edges around the cycle $C_v$ is the same as the order of edges around $v$ and each vertex of $C_v$ has degree 3. Note that the obtained graph $H$ remains planar. For an edge $\{u,v\}\in E(G)$, let $C_v(u)\in C_v$ and $C_u(v)\in C_u$ be its endpoints. For each edge $\{u,v\}\in E(G)$, replace $C_v(u)$ and $C_u(v)$ with paths $C_v(u)_{1,1}...C_v(u)_{d^2,d}$ and $C_u(v)_{1,1}...C_u(v)_{d^2,d}$ such that the following are true (see Figure \ref{fig:scattered}):
		\begin{itemize}
			\item $C_v(u)_{1,1}$ and $C_v(u)_{d^2,d}$ ($C_u(v)_{1,1}$ and $C_u(v)_{d^2, d}$ respectively) are adjacent to exactly one neighbour of $C_v(u)$ in $C_v$ (neighbour of $C_u(v)$ in $C_u$ respectively);
			\item The vertices $C_v(u)_{i,j}$ ($C_u(v)_{i,j}$ respectively) are lexicographically ordered clockwise around the cycle.
		\end{itemize}
		
		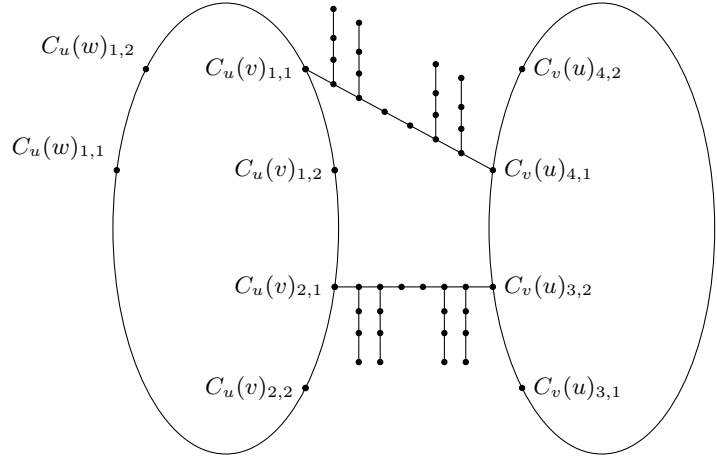
\begin{figure}
			\centering
			\import{figures/}{scattered2.tex}
			\caption{Cycles $C_u$ and $C_v$ and representation of the forbidden pairs $(1,1)$ and $(1,2)$ corresponding to the edge $uv$.}
			\label{fig:scattered}
		\end{figure}
		
		For each edge $\{u,v\}\in E(G)$, each $i\in [d]$ and each color $j\not \in C(u)$, add a path of length $d-1$ to $C_u(v)_{i,j}$. For each edge $\{u,v\}\in E(G)$, consider the forbidden pairs of colors, i.e. the pairs that do not belong to $C(u,v)$. Intuitively, for each such pair $(i,j)$, we select a vertex in $C_u(v)$ which corresponds to color $i$ and a vertex in $C_v(u)$ which corresponds to color $j$. We connect the selected vertices by a gadget which will ensure that we cannot have both of the selected vertices in $S$.
		
		Formally, let $(i,j)\not \in C(u,v)$ and let $a=(i-1)\cdot d+j$. We construct a path $P(u,v,i,j)$ from $C_u(v)_{a, i}$ to $C_v(u)_{d^2+1-a, j}$ of length $2d+1$. Note that the graph remains planar: we have $d^2$ ``blocks'' in $C_u(v)$ (i.e. paths on the cycle which contain one vertex for each color) on both sides, and the $a$th block has an outgoing edge 
		if and only if the $a$th pair (lexicographically ordered) $(i,j)$ does not belong to $C(u,v)$.  
		To each vertex of $P(u,v,i,j)$ except the middle two, add paths of length $d-1$.  
		
		Consider a $d$-scattered set $S$ in $H$. We claim that it corresponds to a solution to the given \textsc{Binary CSP} instance. For any vertex $v$ and a color $i\not\in C(v)$, without loss of generality we may assume that none of the vertices $C_v(u)_{a, i}$ belong to $S$. Also, note that we cannot have $C_v(u)_{a, i}\in S$ and $C_v(u')_{b, j}\in S$ for $i\neq j$. Therefore, there is exactly one color $i$ such that $C_v(u)_{a, i}\in S$ for some $u$, and this will be the color that we will assign to $v$. It remains to check whether the edge constraints are satisfied. Consider an edge $uv\in E(G)$ and its constraints $C(u, v)$. Suppose that in the above process we selected a color $i$ for $u$ and $j$ for $v$, such that $(i,j)\not \in C(u,v)$. Then there exist $a,b\in [d^2]$ such that $C_u(v)_{a, i}, C_v(u)_{b,j}\in S$ and $C_u(v)_{a, i}$ and $C_v(u)_{b, j}$ are connected by a path. Note that on that path no vertex belongs to $S$, which leads to a contradiction.
		
		It remains to bound the outerplanarity, treewidth and pathwidth of $H$. Let $G$ be $k$-outerplanar and let $p(v)$ denote the layer of a vertex $v$, i.e. the number of times the outer face of $G$ needs to be removed in order to have $v$ on the outer face. Note that any vertex of $C_v$ of degree 2 belongs to the same layer as $v$. Therefore, in round $p(v)+1$ all vertices of $C_v$ are on the outer face. In round $p(v)+2$, all vertices of $P(u,v,i,j)$ are on the outer face. Therefore, the graph $H$ is $k+2$-outerplanar.  Note that $k$-outerplanar graphs have treewidth at most $3k-1$ \cite{bodlaender1998partial}.
		
		Let us now compute the pathwidth of $H$. Consider a path decomposition of $G$ of width $k'$. Given a bag $P'_i\subseteq V(G)$, construct $P_i\subseteq V(H)$ as follows. For every $v\in P'_i$, add to $P_i$ all the vertices on $C_v$, and for every edge $\{u,v\}\in E(G)$ such that $u,v\in P'_i$, add all the vertices on paths between $C_u$ and $C_v$, as well as all the vertices on paths with an endpoint at $C_u$ or $C_v$. It is easy to see that the sets $P_i$ form a path decomposition of $H$, and since $d$ is a constant, the resulting path decomposition has width $O(k')$.
	\end{proof}
\section{All-or-Nothing Flow}
\label{sec:flow}
	The \textsc{All-or-Nothing Flow} problem asks if there is an all-or-nothing flow (i.e. a flow where every edge has either zero flow or flow equal to its capacity) with a given value.  

    \defproblem{\textsc{All-or-Nothing Flow}}{A flow network $(G,c,s,t)$, and an integer $r$.}{Is there an all-or-nothing flow from $s$ to $t$
		with value exactly $r$ in $G$?}
        
	In this section, we show that \textsc{All-or-Nothing Flow} is XALP-complete with the outerplanarity
	plus the pathwidth of the graph as parameter. The
	problem is known to be NP-complete \cite{alexandersson2020np}, XNLP-complete with pathwidth
	as parameter~\cite{BodlaenderCW22a} and XALP-complete with treewidth as parameter~\cite{BodlaenderGJPP22a}. The problem is a good starting point for further hardness proofs, and a natural generalisation of network flow.
	
	We remark that the variant where we ask if there is a flow whose value is at least $r$ is
	equally hard, since we can add a new source $s'$ and an arc from $s'$ to the old source $s$
	with capacity $r$.
\begin{theorem}
		\textsc{All-or-Nothing Flow} is XALP-complete parameterized by outerplanarity, for instances with $s$ and $t$ on the same face.
		\label{theorem:aonf-planar-xalp}
\end{theorem}
\begin{proof}
    Since $k$-outerplanar graphs have bounded treewidth, membership in XALP follows from XALP-completeness of \textsc{All-or-Nothing Flow} parameterized by treewidth~\cite{BodlaenderCW22a}. 
	
    By Corollary~\ref{cor:csp_special}, \textsc{Binary CSP} is XALP-complete parameterized by outerplanarity,
		for planar graphs that are connected, bipartite, with each vertex degree two or three, and have disjoint colour sets for adjacent vertices. We reduce from this problem to \textsc{All-or-Nothing Flow} parameterized by outerplanarity.
		
		We employ the technique of
		modelling the choice of a colour for a vertex by choosing a number from a Sidon set, see e.g. \cite{ErdosT41}.
		A \emph{Sidon set} (also known as Golomb ruler) is a set of positive integers $\{a_1, \ldots, a_n\}$ with
		the property that each different pair of integers from the set has a different sum:
		$a_{i_1}+a_{i_2} = a_{j_1}+a_{j_2}$ implies $\{i_1,i_2\}=\{j_1,j_2\}$. 
		Erd\H{o}s and Tur\'{a}n~\cite{ErdosT41} have shown that for each $n$, there
		is a Sidon set of size $n$ with all numbers in the set at most $4n^2$; one can also observe
		from their proof that this set can be constructed with logarithmic space.
		
		Suppose we are given an instance of \textsc{Binary CSP} with outerplanarity as parameter for bipartite planar graphs with all vertices degree two or three.
		In other words, we are given a graph $G=(V,E)$, a set $D$ of colours, for each vertex $v\in V$,
		a subset $C(v)\subseteq D$, and for each ordered pair $(v, w)$ such that $\{v, w\}\in E$,
		a set of allowed colour pairs $C(v,w)\subseteq C(v)\times C(w)$. Without loss of generality, we may assume that the sets $\{C(v):\: v\in V\}$ partition $D$.
		
    Let $\{a_1,
		\dots, a_{|D|}\}$ be a Sidon set such that all its elements are positive integers in $O(|D|^2)$.
		 Let $L = \max \{a_1, \ldots, a_{|D|}\}+1$. Note
		that the set $\{a_1+L, \ldots, a_{|D|}+L\}$ is also a Sidon set.
		We assign to each colour $c$ in $D$ a unique element $s(c)$ from $\{a_1+L, \ldots, a_{|D|}+L\}$.
		
		For a vertex $v\in V(G)$, we write $S(v) = \{s(c):\:  c\in C(v)\}$ for the set of Sidon numbers
		of the colours we can give to $v$. 
		We define $S(v)+1 = \{ s(c)+1:\:  c\in C(v)\}$ as the set obtained by adding one
		to each element from $S(v)$.
		For a pair of vertices $v,w\in V(G)$, we write
		$S(v)+S(w) = \{ s(c)+s(c') :\: c\in C(v) \wedge c'\in C(w) \}$.
		For an ordered pair of endpoints of an edge $\{v,w\}\in E$,
		we write $S(v, w) = \{ s(c)+s(c') :\: c\in C(v) \wedge c'\in C(w) \wedge (c,c')\in C(v, w)\}$.
		
		\paragraph{Gadgets.}
		Before describing the construction of $H$, we introduce three similar auxiliary gadgets, 
		namely the XOR-gadget for vertices, the XOR-gadget for edges, and
		an $s-t$-gadget. 
		
		In order to simplify the drawing, instead of parallel arcs, we will draw one arc whose capacity is denoted by the (multi)set of capacities of the parallel arcs. We denote the multiset consisting of $t$ ones by $1_t$.
		
		Informally, given a vertex $v\in V(G)$, the XOR-gadget for vertex $v$, denoted by $X(v)$ in the illustrations
		will ``select'' exactly one value $s(c)\in S(v)$, which corresponds to coloring $v$ with colour $c$. Formally, given a vertex $v\in V(G)$, $X(v)$ is a graph that consists of three nodes, $v_1, v_2, v_3$ and has $2L$ arcs of capacity 1 from $v_1$ to $v_2$ and for each element of $S(v)$ an arc from $v_2$ to $v_3$ with that capacity (see Figure~\ref{fig:aon-gadget} (a)). Note that, since all elements of $S(v)$ are between $L$ and $2L-1$ and the incoming flow to $v_2$ is at most $2L$, we cannot have more than one arc from $v_2$ to $v_3$ with nonzero flow. 
		
		The XOR-gadget for an edge $\{v,w\}\in E$ is similar. Here, we have three
		vertices $v_1,v_2,v_3$, with $4L$ arcs with capacity 1 from $v_1$ to $v_2$,
		and for each element from $S(v,w)$ an arc from $v_2$ to $v_3$ with that capacity. By the above argument, we cannot have more than one arc from $v_2$ to $v_3$ with nonzero flow. See Figure ~\ref{fig:aon-gadget} (b).
		
		We will also construct one $s-t$-gadget. Fix an arbitrary vertex $v'\in V(G)$. 
		The $s-t$-gadget is obtained by taking an XOR-gadget for $v'$,
		and adding one to the capacity of each arc from $v'_2$ to $v'_3$. (see
		Figure~\ref{fig:aon-gadget} (c)).
		
		\begin{figure}
			\centering
			\includegraphics[width=0.8\linewidth]{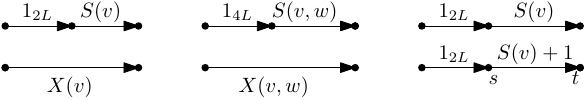}
			\caption{(a) The XOR-gadget for a vertex $v$ and its representation. (b) The XOR-gadget for an edge $\{v,w\}$ and its representation. (c) Changing a XOR-gadget to the $s$-$t$-gadget. $S(v)+1$ denotes the set $\{x+1~|~x\in S(v)\}$.}
			\label{fig:aon-gadget}
		\end{figure}
		
		Let $V=V_1 \cup V_2$ with edges only between $V_1$ and $V_2$.
		
		For each vertex $v\in V_1$, we do the following. Suppose that the neighbours of $v$ are $w_1,\dots, w_t$ in clockwise order. We construct a clockwise cycle that alternates between an XOR-gadget for $v$ and an XOR-gadget for $\{v, w_i\}$ (see Figure~\ref{fig:aon2}).
		
		For example, if the incident edges, in clockwise order of $v \in V_1$ are $\{v,w\}$, $\{v,y\}$,
		$\{v,x\}$, then we have, in clockwise order:
		a XOR-gadget for $v$,
		a XOR-gadget for $\{v,w\}$,
		a XOR-gadget for $v$,
		a XOR-gadget for $\{v,x\}$,
		a XOR-gadget for $v$,
		a XOR-gadget for $\{v,y\}$.
		
		For each vertex in $V_2$ we have a cycle, that is embedded counter clockwise.
		The construction is the same as for vertices in $V_1$, except that we go counter clockwise.
		
		Now, in the final step, we turn one XOR-gadget for $v'$ into an $s$-$t$-gadget as discussed above. We call the graph obtained by this procedure $H$.
		
		In order to show that the graph $H$ is planar, it suffices to show the planarity of the graph $\tilde{H}$, obtained from $H$ by removing parallel edges. 
		Note that each vertex in $\tilde{H}$ has degree 3, so $\tilde{H}$ cannot have a $K_5$-subdivision. Suppose that $\tilde{H}$ contains a $K_{3,3}$-subdivision $\tilde{H}'$. Since each vertex in $\tilde{H}'$ has degree 3, $\tilde{H}'$ must contain all edges in $\tilde{H}$ incident to this vertex - in particular, it will contain edges from the edge gadget. If we consider the subgraph of $G$ obtained by taking the corresponding edges, we obtain a $K_{3,3}$-subdivision, which leads to a contradiction.  
		
		
		
		\begin{figure}
			\centering
			\includegraphics[width=\textwidth]{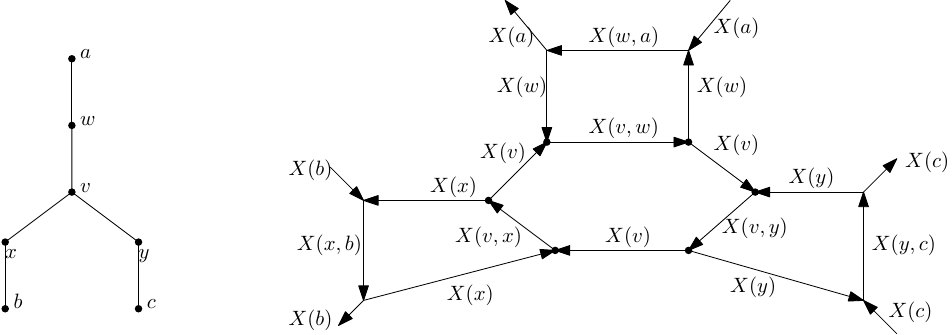}
			\caption{The construction for a vertex with neighbours}
			\label{fig:aon2}
		\end{figure}

		Let $H$ be the resulting graph. Parallel edges can be avoided by subdividing
		each arc into two arcs of the same capacity. Clearly, this gives
		a graph without parallel arcs that forms an equivalent instance.

		\begin{claim}
			There is an all-or-nothing flow in $H$ from $s$ to $t$ of value exactly $1$ if and only if $G$ has a coloring that
			satisfies all constraints on vertices and edges.
		\end{claim}
		
		\begin{claimproof}
			Suppose we have a coloring $c: V\rightarrow \mathcal{C}$ of $G$ that satisfies
			all constraints on vertices and edges. 
			Now, for each $v$ and each $X(v)$ gadget, we send $s(c(v))$ flow through
			the gadget --- we send $1$ flow through $s(c(v))$ many edges of capacity 1,
			and send $s(c(v))$ flow over the edge of capacity $s(c(v))$, and no flow
			through the other edges with capacities larger than one in the gadget.
			For each edge $e = vw$, we send $s(c(v))+s(c(w))$ flow over the edge with
			that capacity from
			pair $x_{in, e}$ to $x_{out,e}$. 
		Suppose that the $s-t$ gadget is a modification of the XOR-gadger for $v'$.	
            In the $s$-$t$ gadget, we do the same as for a standard vertex gadget, except that we send
            $s(c(v'))+1$ flow from $s$ to $t$. Now, all vertices except $s$ and $t$ have their inflow
            equal to their outflow, while the outflow of $s$ is one larger than its inflow, and the
            inflow of $t$ is one larger than its outflow. So, we have an all-or-nothing flow from $s$ to
            $t$ of value exactly 1.
			
			For the other direction, suppose that we have an all-or-nothing flow in $H$ of value 1. Suppose that the $s-t$-gadget is a modification of the XOR-gadget for $v'$. Note that at least one edge with $s$ as tail must have
            non-zero flow, otherwise the flow value is 0. As the inflow to $s$ is at most $2L$, also at most one
            edge from $s$ has non-zero flow. 

            We now claim that for each vertex $w$, the through all $X(w)$ gadgets is equal, and is in $S(w)$. A special
            case is the gadget that contains $s$ and $t$: the flow through the gadget is in $S(v')$, but the edge from
            $s$ to $t$ with non-zero flow has a capacity that is one larger --- this gives the flow value of 1.
            
 We will prove this claim by induction on the distance between $v'$ and $w$ in $G$. For the base case, let $w$ be the vertex such that one endpoint of an XOR-gadget for $w$ is $t$. The flow through the $X(v, w)$ gadget must belong to $S(v,w)$. The outgoing flow of this gadget splits between an $X(w)$-gadget and an $X(v)$-gadget. Since we know that $S(v)$ and $S(w)$ are disjoint subsets of a Sidon set, the outgoing flow of $X(v,w)$ can split in only one way, namely the $X(v)$-gadget gets incoming flow $r$, and the $X(w)$-gadget gets incoming flow $r'\in S(w)$. By repeatedly applying the above argument, we conclude that all the gadgets $X(v)$ have the flow equal to $r$. The induction step is analogous to the base case. 
			
			It is easy to see that the edge constraints are satisfied: namely, consider an edge $uu'\in E(G)$ and let $f$ be the flow through the gadget $X(u,u')$. 
            Since all elements of $S(u, u')$ are between $2L+2$ and $4L$, the flow $f$ has to use exactly one of the parallel edges with capacities from $S(u, u')$.
            Since $f\in S(u, u')$, we know that $f=s(c)+s(c')$ for colors $c\in C(u)$ and $c'\in C(u')$ such that $(c,c')\in C(u,u')$.  
		\end{claimproof}
        
        It remains to show that the graph $H$ has bounded outerplanarity. Let $v\in V(G)$ be a vertex that is on the outer face after $t$ rounds. We will show by induction on $t$ that the start and endpoints of all $X(v)$ gadgets are on the outer face after at most $2t$ rounds. 
        
        Let $v'\in V(G)$ be a vertex that shares a face with $v$ and is on the outer face after $t-1$ rounds. By induction, all the start and endpoints of all $X(v')$ gadgets are on the outer face after at most $2t-2$ rounds. In particular, the $X(v)$ gadget that shares a face with a $X(v')$ gadget will be on the outer face after at most $2t-1$ rounds. Thus all the start and endpoints of $X(v)$ gadgets will be on the outer face after at most $2t$ rounds. 
	\end{proof}
		
		Let us now prove XALP-completeness for \textsc{All-or-Nothing Flow} for a restricted set of input instances, that we will call \emph{special instances}. The underlying graph of a special instance is planar. The source $s$ and sink $t$ belong to the same face. For each vertex $v$ that is not equal to $s$ or $t$, the total capacity of the edges to
        $v$ equals the total capacity of edges out of $v$.

		\begin{lemma}
			\textsc{All-or-Nothing Flow} is XALP-complete parameterized by outerplanarity for special instances.
            \label{lemma:aonf-specialinstances}
		\end{lemma}
		
		\begin{proof}
		Consider the graph $H$ from the proof of Theorem~\ref{theorem:aonf-planar-xalp} (for simplicity, consider the version with parallel edges). We will construct an equivalent special instance $H'$ by taking the graph $H$ and adding arcs.
        Let $Q = \sum_{c\in D} s(c)$.
			We add arcs to $H'$ in the opposite direction of arcs in $H$.
			For a pair of vertices $v$, $w$ with at least one arc $vw$ in $H$, we have the following cases:
			\begin{itemize}
				\item If we have arcs of the form $1_{2L}$, respectively $1_{4L}$ from $v$ to $w$, then add an arc from $w$ to $v$ with capacity $2L+Q$, respectively $4L+2Q$.
				\item If we have a collection of arcs with capacities $S(x)$ for some vertex $x$ from $G$, then add an arc from $w$ to $v$ with capacity $Q+ \sum_{c\in C(x)} s(c)$.
				\item If we have a collection of arcs with capacities $S(x,y)$ for some edge $\{x,y\}$ from $G$, then add an arc from $w$ to $v$ with capacity $2Q+ \sum_{(c,c')\in C(x,y)} (s(c)+s(c'))$.
			\end{itemize}

With a simple case analysis, we can verify that the resulting graph $H'$ is special. 
			
			Clearly, if there is a flow in $H$, then there is a flow of the same value in $H'$ (by not using any of the new arcs). In the reverse direction, we notice that a flow with
			positive value from $s$ to $t$ cannot use any of the new arcs. Indeed, if it uses some arc
			with capacity larger than $Q$, then all incident arcs with capacity larger than $Q$ 
			also must be used, or we violate the flow conservation law. But this implies that
			the flow from $s$ to $t$ is larger than $Q$. Thus the all-or-nothing flow with value exactly 1 in the new network is also an all-or-nothing flow in the old network.
		\end{proof}

\section{Reductions from All-or-Nothing Flow}
	\label{sec:reductionsfromallornothing}
	The \textsc{All-or-Nothing Flow} problem appears to be a good starting point for many other hardness proofs. In all cases, the
	reductions replace edges or arcs by a small gadget with outerplanarity at most 2, and thus the
	graph stays planar and the outerplanarity increases by at most 2. This way, XALP-completeness is obtained 
	for each of the following problems on planar graphs with outerplanarity as parameter. Values are always
	given in unary. 
	
	\label{sec:reductions_from_flow}
	In this section, we build up on results of Section~\ref{sec:flow} and prove hardness of several problems. Figure~\ref{fig:reductions} shows the reductions used in this section. 
	
	\begin{figure}[ht]
		\import{figures/}{reductions.tex}
		\caption{Reductions from \textsc{All-or-Nothing Flow}.}
		\label{fig:reductions}
	\end{figure}
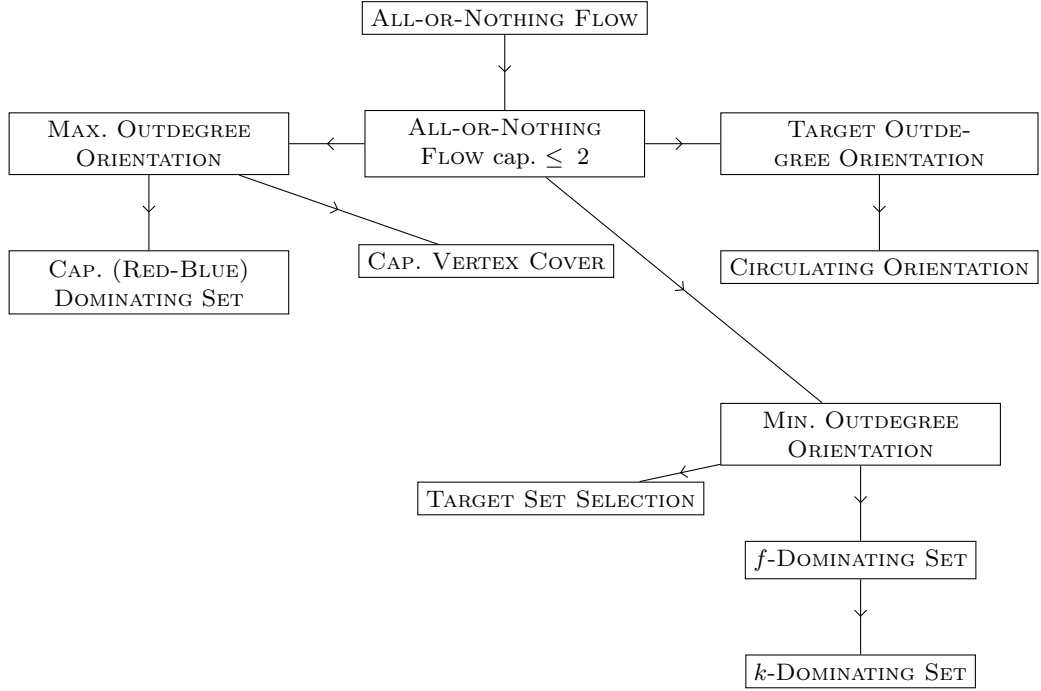
	
	\subsection{All-or-Nothing Flow with Small Arc Capacities}
	\label{ssec:flow12}
	We now look at the hardness of \textsc{All-or-Nothing Flow} for the case that all arc capacities 
	are small. Note that when all arc capacities are 1, then a maximum all-or-nothing flow has
	the same value as a maximum flow, and we can use a standard flow algorithm like the 
	Ford-Fulkerson algorithm to solve the problem in polynomial time. The next case, where capacities
	are 1 and 2 is already equally hard as the case where capacities are given in unary.
	
	We will see that a simple transformation can transform an arc to an equivalent subgraph
	which is 2-outerplanar, and with all arc capacities 1 or 2. As the transformation to
	\textsc{Target Outdegree Orientation} given in the next section
	has edge weights that equal the arc capacities, the hardness results carry
	over to that problem.
	
	Suppose $xy$ is an arc with integer capacity $\gamma\geq 3$. 
	We replace the arc $xy$ by the following subgraph (see Figure~\ref{fig:weighted12too}):
	\begin{itemize}
		\item Take a directed path with $2\gamma-4$ new vertices, say $v_1, v_2, \ldots, v_{2\gamma}-4$.
		\item If $i \in [1,2\gamma-5]$ is odd, then the arc $v_iv_{i+1}$ has capacity 1.
		\item If $i \in [2,2\gamma-6]$ is even, then the arc $v_iv_{i+1}$ has capacity 1.
		\item Take an arc from $x$ to $v_1$ of capacity 2.
		\item For each even $i \in [2,2\gamma-4]$, take an arc from $x$ to $v_i$ of capacity 1.
		\item For each odd $i \in [1,2\gamma-5]$, take an arc from $v_i$ to $y$ of capacity 1.
		\item Take an arc from $v_{2\gamma-4}$ to $y$ of capacity 2.
	\end{itemize}
	
	\begin{lemma}
		Suppose we have a subgraph, as described above, with only $x$ and $y$ adjacent to vertices
		outside the gadget.
		Then for each all-or-nothing flow, either all arcs in the gadgets
		have flow 0, or all arcs in the gadget have flows equal to their capacity.
	\end{lemma}
	
	\begin{proof}
		Each vertex $v_i$ in the gadget
		has either two incoming arcs of capacity 1 and one outgoing arc of
		capacity 2, or one incoming arc of capacity 2 and two outgoing arcs of capacity 1. So either
		all three arcs incident to $v_i$ are used, or all three have flow 0. 
		Suppose for contradiction that the claim from the lemma does not hold.    Then there must be a vertex $v_i$ with inflow and outflow 0, and a vertex $v_{i'}$ with
		inflow and outflow 2. Moreover, there must be such a pair that is adjacent ($|i-i'|=1$),
		which gives a contradiction.
	\end{proof}

	\begin{figure}[ht]
		\centering
		\includegraphics{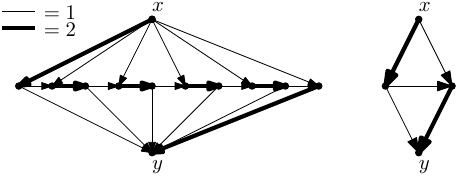}
		\caption{Examples of the gadget for arcs with capacities $>2$. Fat arcs have capacity 2; thin arcs have capacity 1. The left gadget in the example functions as an arc with capacity 7; the right for
			an arc with capacity 3.}
		\label{fig:weighted12too}
	\end{figure}
	
	\begin{corollary}
		The \textsc{All-or-Nothing Flow} problem with all arcs of capacity 1 or 2 is:
		\begin{enumerate}
			\item XALP-complete for planar graphs with outerplanarity as parameter.
			\item XNLP-complete with pathwidth as parameter.
			\item XALP-complete with treewidth as parameter.
		\end{enumerate}
		\label{corollary:aon12}
	\end{corollary}
	
	\begin{proof}
		For each of the hardness proofs, we replace each arc with capacity more than 2 by the gadget
		described above. Notice that this step increases the outerplanarity of a planar
		graph by at most 1 (the vertices $v_i$ are at most one level more away from the outer face
		than $x$ or $y$) and the pathwidth of a graph by at most 2 (take a bag containing $x$ and $y$, replace it by $2\gamma-5$ copies, and add $v_i$, $v_{i+1}$ to the $i$th copy). 
		Also, the operation does not increase
		the treewidth of a graph whose treewidth is at least three: the gadget has treewidth at most 3, and
		we can attach a tree decomposition of a gadget to a bag that contains $x$ and $y$.
	\end{proof}
	
	\subsection{Target Outdegree Orientation}
	\label{ssec:too}
	In this section, we look at the problems \textsc{Target Outdegree Orientation}, 
	\textsc{Maximum Outdegree Orientation} and \textsc{Minimum Outdegree Orientation}, and show that they
	are XALP-complete when parameterized by outerplanarity. Given a directed graph $G=(V,E)$ with
	edge weights $w(e)$, the \emph{weighted outdegree} of a vertex is the sum of the weights
	of all arcs with $v$ as tail: $\sum_{vx\in E} w(vx)$.
	
	\defproblem{\textsc{Target Outdegree Orientation}}{Undirected graph $G=(V,E)$, a positive integer weight $w(e)$ for each edge $e\in E$, and for each vertex $v\in V$, a 
		positive integer target $t(v)$.}{Is there an orientation of $G$, such that for each $v$, the weighted outdegree of $v$ equals $t(v)$?}
	
	The \textsc{Minimum Outdegree Orientation}
	and \textsc{Maximum Outdegree Orientation} problem are defined as above, except that now the weighted
	outdegree of each vertex $v$ must be at least $t(v)$, respectively at most $t(v)$.
	Each of these three problems was shown to be XNLP-complete with pathwidth as parameter~\cite{BodlaenderCW22a}, FPT with tree partition width as parameter~\cite{BodlaenderCW22a}, and XALP-complete with treewidth as parameter~\cite{BodlaenderGJPP22a}.
	
	To show hardness of each of the three problems, we use a simple transformation
	from \textsc{All-or-Nothing Flow}: drop directions of all arcs,
	take weights equal to capacities, and
	choose targets for vertices in an appropriate way. Our transformation is similar to,
	but somewhat simpler than the transformation given in \cite{BodlaenderCW22a}. 
	
	\smallskip
	
	Suppose we have a flow network $(G, s, t, c)$ with $G=(V,E)$ a directed graph,
	$c: E \rightarrow \mathbb{Z}^+$ the capacity function, $s,t\in V$. Suppose we want to
	decide if there is a flow from $s$ to $t$ of value exactly $r$.
	
	We assume that there are no parallel arcs, and also for each pair of vertices $x,y\in V$,
	at most one of the arcs $xy$ and $yx$ is in $E$; if not, we can obtain an equivalent instance
	by subdividing arcs and giving both new arcs the same capacity as the old arc.
	
	Let $H=(V,F)$ be the undirected graph, obtained by dropping directions of arcs in $E$,
	and let the weight $w(e)$ of an edge be equal to the capacity of the directed variant of $e$ in $G$.
	
	For each vertex $v\in V\setminus\{s,t\}$, set $t(v) = \sum_{xv\in E} c(xv)$.
	Set $t(s)=\sum_{xs\in E} c(xs) + r$, and $t(t)= \sum_{xt\in E} c(xt) - r$.
	
	\begin{lemma}
		Let $G$ and $H$ be as above. The following are equivalent.
		\begin{enumerate}
			\item There is an all-or-nothing flow of value $r$ in $G$.
			\item $H$ has on orientation with
			each vertex weighted outdegree exactly $t(v)$.
			\item $H$ has on orientation with
			each vertex weighted outdegree at least $t(v)$.
			\item $H$ has on orientation with
			each vertex weighted outdegree at most $t(v)$.
		\end{enumerate}
	\end{lemma}
	
	\begin{proof}
		1 $\Rightarrow 2$:
		Consider a flow $f$. If an arc $xy$ has positive flow, then we direct the edge $\{x,y\}$ in
		$H$ in the same way as the arc in $G$, i.e. from $x$ to $y$. If $xy$ has 0 flow, then
		we direct the edge in the opposite direction, i.e. from $y$ to $x$.
		
		Consider a vertex $v\in V\setminus \{s,t\}$. 
		Suppose the total flow sent to $v$ by $f$ (i.e. the \emph{inflow} $\sum_{xv\in E} f(xv)$) is $\alpha$.
		Then, $t(v)-\alpha$ weight of incoming edges is not used; thus, the incoming arcs in $G$
		give weighted outdegree $t(v)-\alpha$ in $H$. The total weight of used outgoing arcs from $v$ is $\alpha$; these are directed out of $v$, while all other outgoing arcs are directed in the opposite direction; this gives a contribution of $\alpha$ to the outdegree.
		So, the total outdegree equals $t(v)-\alpha+\alpha=t(v)$.
		
		If the inflow of $s$ by $f$ is $\alpha$, then the outflow of $s$ by $f$ is $\alpha+r$;
		we thus have $t(s)-r-\alpha$ weight of incoming edges not used by $f$, thus incoming arcs
		amount 
		to weighted outdegree $t(s)-r-\alpha$; 
		The analysis for $t$ is similar.
		
		2 $\Rightarrow$ 1:
		Suppose we have an orientation with each vertex meeting its target. For each arc $xy$,
		send flow equal to its capacity over it,  if and only if the edge $\{x,y\}$ is oriented 
		from $x$ to $y$ (i.e., in the same way as the corresponding arc in $H$); otherwise we
		send 0 flow over the arc.
		Consider a vertex $v$.
		Suppose the total
		weight of arcs $xv$ whose edges are directed as $xv$ is $\beta$. Then, the weight of
		the edges with an incoming arc to $v$ which are directed out of $v$ is $t(v)-\beta$, so
		the weight of the edges $vy$ with an outgoing arc from $v$ who are directed as $vy$ is
		$\beta$; so the inflow and outflow of $v$ equal $\beta$.
		The analysis for $s$ and $t$ is similar.
		
		2 $\Rightarrow$ 3, 4 is trivial: use the same orientation.
		
		3 $\Rightarrow$ 2: we can use the same orientation. Note that the sum over all vertices of the weighted outdegree
		must equal the sum of the weights of all edges (as each edge is directed out of exactly one vertex in an orientation), and the latter sum equals $\sum_{v\in V} t(v)$. Thus
		if we have a vertex whose weighted outdegree is strictly greater than its target, then there must be
		another vertex whose weighted outdegree is strictly smaller than its target, which leads to a contradiction.
		
		4 $\Rightarrow$ 2 is similar as the previous case.
	\end{proof}

	\begin{corollary}
		The \textsc{Target Outdegree Orientation}, 
		\textsc{Minimum Outdegree Orientation}, and 
		\textsc{Maximum Outdegree Orientation} problems with all arcs of capacity 1 or 2 is:
		\begin{enumerate}
			\item XALP-complete for planar graphs with outerplanarity as parameter.
			\item XNLP-complete with pathwidth as parameter.
			\item XALP-complete with treewidth as parameter.
		\end{enumerate}
	\end{corollary}
	
	\begin{proof}
		Start with an instance of \textsc{All-or-Nothing Flow} with all capacities 1 or 2, see
		Corollary~\ref{corollary:aon12}. Then, set the targets as described above and drop the directions of arcs.
	\end{proof}

 \subsection{Circulating Orientation}
		\label{ssec:co}
		The \textsc{Circulating Orientation} problem is defined as follows. 

        \defproblem{\textsc{Circulating Orientation}}{Undirected graph $G=(V, E)$ and a positive integer weight $w(e)$ for each edge $e\in E$}{Is there an orientation of $G$ such that for each $v\in V$ the weighted outdegree of $v$ equals the weighted indegree of $v$?}

        This problem was shown to be XALP-complete parameterized by treewidth~\cite{BodlaenderGJPP22a} and XNLP-complete parameterized by pathwidth~\cite{BodlaenderCW22a}. The problem is also known as \textsc{Flow Orientation} and it is NP-hard on planar graphs~\cite{didimo2019hv}. In~\cite{JansenKKLMS23}, it was shown that the problem is W[1]-hard.
        
		\begin{theorem}
			\textsc{Circulating Orientation} is XALP-complete with outerplanarity as parameter.
		\end{theorem}
		
		\begin{proof}
			Consider the reduction from \cite{BodlaenderCW22a} from \textsc{All-or-Nothing Flow} to \textsc{Circulating Orientation}, building upon
			insights by Jansen et al.~\cite{JansenKKLMS23}. 
			
			Take a special instance of \textsc{All-or-Nothing Flow}, as in Lemma~\ref{lemma:aonf-specialinstances}.
			
			The first step makes an instance of \textsc{Target Outdegree Orientation}.
			Here, we first transform each arc to an undirected edge with weight equal to half the capacity of the arc. Then,
			each vertex is given a target value. For all vertices, except source $s$ and
			sink $t$, this target value is the sum of the capacities of the arcs directed
			to $v$, divided by 2. 
			Consider a vertex $v \in V\setminus \{s,t\}$. The target value is equal to
			half the sum of the capacities of the arcs directed to $v$, but, in our current
			instance, this is also equal to half the sum of the capacities of the arcs
			directed out of $v$. This means that in the instance, we look for an
			orientation with a target that is equal to half the total weight of the incident edges.
			This holds for all vertices, except $s$ and $t$, and these two vertices
			share a face in the plane embedding.
			
			In the second step, we transform the instance of \textsc{Target Outdegree Orientation} to an instance of \textsc{Circulating Orientation}.
			Here, we add new `super source vertex', say $s'$, and a new super sink
			vertex, say $t'$. For each vertex, we compute its demand, which is the
			total weight of incident edges minus twice its target. Now, from the discussion
			above, we see that the demand is 0 for all vertices except $s$ and $t$.
			
			In the final step, we add edges with positive weights from $s'$ to $t'$,
			from $s'$ to each vertex with a negative demand, and from $t'$ to each vertex
			with a positive demand. Because $s$ and $t$ are the only two vertices
			which possibly have a non-negative demand, and they are on the same face, we
			can add these edges keeping a planar graph. This step increases the outerplanarity by at most one.
		\end{proof}
    
	\subsection{Capacitated (Red-Blue) Dominating Set}
	\label{ssec:capds}
	We can use the \textsc{Minimum Outdegree Orientation} and
	\textsc{Maximum Outdegree Orientation} problems as starting problems for 
	reductions to several problems, including \textsc{Capacitated Red-Blue Dominating Set},
	\textsc{Capacitated Dominating Set}, \textsc{Capacitated Vertex Cover}m
	\textsc{Target Set Selection}, and \textsc{$f$-Dominating Set}.
	In each of these cases, we take an instance of \textsc{Minimum Outdegree Orientation} 
	or \textsc{Maximum Outdegree Orientation}
	and replace each edge by an appropriate planar subgraph.
	We start with the capacitated variants of \textsc{Dominating Set} and \textsc{Vertex Cover},
	which use existing transformations from the literature.
	
	\defproblem{\textsc{Capacitated Red-Blue Dominating Set}}{Bipartite graph $G=(R\cup B, E)$, with each `blue' vertex in $v\in B$ a positive integer
		capacity $c(v)$, and an integer $k$}{Is there a set $S\subseteq B$ of at
		most $k$ blue vertices, and an assignment $f: R\rightarrow S$ of each red vertex to
		a (blue) neighbour in $S$, such that each vertex in $v\in S$ has
		at most $c(v)$ red neighbours
		assigned to it?}
	
	\defproblem{\textsc{Capacitated Dominating Set}}{Graph $G=(V, E)$, with each vertex in $v\in R$ a positive integer
		capacity $c(v)$, and an integer $k$}{Is there a set $S\subseteq V$ of at
		most $k$ vertices, and an assignment $f: V\setminus S\rightarrow S$ of each vertex not in $S$ to
		a neighbour in $S$, such that each vertex in $v\in S$ has
		at most $c(v)$ neighbours
		assigned to it?}
	
	\textsc{Capacitated Red-Blue Dominating Set} and \textsc{Capacitated Dominating Set} are
	well studied problems. Both problems are XNLP-complete with pathwidth as parameter~\cite{BodlaenderGJJL22} and XALP-complete with treewidth as parameter~\cite{BodlaenderGJPP22a}. 
	An earlier W[1]-hardness proof for treewidth as parameter can
	be found in \cite{DomLSV08}.
	\textsc{Capacitated Dominating Set} was shown to be W[1]-hard for planar graphs with the solution size as parameter~\cite{BodlaenderLP09}. It follows
	from a lower bound proof in ~\cite{FominGLS14} that \textsc{Capacitated Red-Blue Dominating Set}
	is W[1]-hard with feedback vertex set as parameter.

	In \cite{BodlaenderGJJL22}, the following transformation from  \textsc{Maximum Outdegree Orientation} to
	\textsc{Capacitated Red-Blue Dominating Set} is given: each edge $e=\{v,w\}$ with capacity $\gamma$ is
	replaced by a subgraph with $2\gamma+3$ additional vertices. Take $\gamma$ red vertices adjacent to $v$ and
	$\gamma$ red vertices adjacent to $w$. Take a new blue vertex adjacent to the each vertex in the
	first set of $\gamma$ red vertices, and a new blue vertex adjacent to the each vertex in the
	second set of $\gamma$ red vertices. Then take one red vertex, adjacent to each of the latter two
	blue vertices. Finally, we colour all original vertices blue. The blue vertices in the edge gadgets
	have capacity $\gamma+1$; the original vertices $v\in V$ have capacity $t(v)$. 
	We ask if there is a red-blue dominating set of size $|V|+|E|$. See Figure~\ref{fig:capredblue} (left side) for an example.
	
	Correctness of the transformation
	is easy to see (cf.~\cite{BodlaenderGJJL22}); the intuition is that there is always an optimal
	solution $S$ that contains all original vertices in $V$, and for each edge gadget, one of the blue
	vertices (marked $a$ and $b$ in Figure~\ref{fig:capredblue}) is chosen. If $a$ is chosen, we direct
	the edge from $y$ to $x$; if $b$ is chosen, the edge is directed from $x$ to $y$. A vertex dominates
	all neighbouring vertices in edge gadgets of outgoing edges - a solution of the \textsc{Maximum Outdegree Orientation} thus has the property that original vertices $v\in V$ dominate at most $t(v)$ neighbours.
	
	Transforming to \textsc{Capacitated Dominating Set} is done by dropping all colours, and giving
	each vertex that was blue an adjacent $P_2$, with the first vertex on this $P_2$ having weight 2.
	All vertices that were originally red have capacity 1.
	On each of these $P_2$s, one vertex must be placed in
	the solution; we can choose the vertex incident to the originally blue vertex, which thus is dominated.
	See \cite{BodlaenderGJJL22}. An example can be found in Figure~\ref{fig:capredblue}
	(right side).
	
	\begin{figure}
		\centering
		\includegraphics{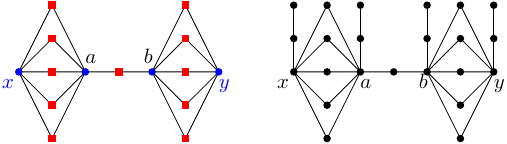}
		\caption{Left: Example of edge transformation for \textsc{Capacitated Red-Blue Dominating Set}. 
			Right: Example of edge transformation for \textsc{Capacitated Dominating Set}.
			The edge has weight $5$, $a$ and $b$ have capacity $6$. $x$ has capacity $t(x)$; $y$ has capacity $t(y)$. All other black vertices (right figure) have capacity 1. }
		\label{fig:capredblue}
	\end{figure}
	
	Note that the transformations keep the planarity of the graph invariant, and thus we can conclude:
	
	\begin{theorem}
		\textsc{Capacitated Red-Blue Dominating Set} and \textsc{Capacitated Dominating Set}
		are XALP-complete for planar graphs with outerplanarity as parameter.
	\end{theorem}
	
	\subsection{Capacitated Vertex Cover}
	\label{ssec:capvc}
	Another well studied `capacitated' variant of a classic graph problem is 
	\textsc{Capacitated Vertex Cover}. It is XNLP-complete for pathwidth (see \cite{BodlaenderGJJL22}),
	XALP-complete for treewidth~\cite{BodlaenderGJJL22}, improving
	upon an earlier W[1]-hardness proof~\cite{DomLSV08}. Dom et al.~\cite{DomLSV08} give an $O(2^{O(tw \log k)}n^{O(1)})$ time
	algorithm with $k$ the solution size and $tw$ the treewidth and thus show the problem to be FPT for the combined parameter of treewidth and solution size.
	
	\defproblem{\textsc{Capacitated Vertex Cover}}{Graph $G=(V,E)$, positive integer capacity $c(v)$ for each $v\in V$, integer $k$}{Is there a set of vertices $S\subseteq V$, with $|S|\leq k$, and 
		an assignment of each edge to an incident vertex in $S$ such that
		no vertex $v\in S$ has more than $c(v)$ edges assigned to it?}
	
	In~\cite{BodlaenderGJJL22}, the following reduction from \textsc{Maximum Outdegree Orientation} to \textsc{Capacitated Vertex Cover} is discussed: replace each edge $\{x,y\}$ of weight
	$\gamma$ by
	a gadget, with new vertices $a$, $b$, $\gamma$ paths of length 3 from $x$ to $a$, one edge from $a$ to $b$, and $\gamma$ paths of length 3 from $b$ to $y$. For each original vertex $v\in V$,
	add one pendant neighbour. Set for all new vertices their capacity
	equal to their degree, and for the original vertices $v\in V$,
	their capacity equal to their target outdegree plus 1. See Figure~\ref{fig:capvc} for an example.
	
	The correctness proof of the transformation is similar to that
	for \textsc{Capacitated Red-Blue Dominating Set} but was not given
	in \cite{BodlaenderGJJL22}.
	
	\begin{figure}[ht]
		\centering
		\includegraphics[width=\textwidth]{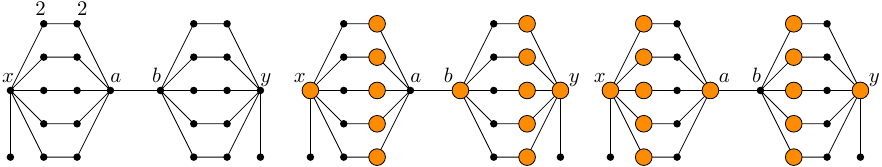}
		\caption{Edge gadget for \textsc{Capacitated Vertex Cover}.
			If the edge $\{x,y\}$ has weight $\gamma$, then $c(a)=c(b)=\gamma+1$. The capacity of the degree-2 vertices in the gadget is 2. $c(x)=t(x)+1$;
			$c(y)=t(y)+1$. Middle: corresponds to directing the edge from
			$x$ to $y$. Right: corresponds to directing the edge from $y$ to $x$.}
		\label{fig:capvc}
	\end{figure}
	
	\begin{lemma}
		Let $H$ be the graph obtained from $G$ by applying the
		reduction as described above to each edge of $G$, with the
		given capacities. $G$ has an orientation with each vertex $v\in V$
		weighted outdegree at most $t(v)$, if and only if 
		$H$ has a capacitated vertex cover of size at most $|V|+\sum_{e\in E} (2w(e)+1)$. 
	\end{lemma}
	
	\begin{proof}
		Suppose we have an orientation of $G$ with each vertex
		outdegree at most its target. Take the following set $S$:
		place all original vertices (in $V$) in $S$, and for
		each edge $xy$, if it is directed from $x$ to $y$, all vertices
		in the gadget with even distance to $x$, and if it is directed
		from $y$ to $x$, all vertices with odd distance to $y$, see
		Figure~\ref{fig:capvc}. 
		
		All vertices inside the gadget have capacity equal to their degree, and when in $S$, get all incident edges assigned to it.
		A vertex $x\in V$ has assigned to it its incident vertex of degree one, and the neighbours in edge gadgets of outgoing edges. Note that the number of the latter type of neighbours equals the weighted outdegree of $v$ in the orientation.
		Thus, we have a capacitated vertex cover of the correct size.
		
		\smallskip
		Now, suppose $H$ has a capacitated vertex cover $S$ of size
		at most $|V|+\sum_{e\in E} (2w(e)+1)$. For each $x\in V$,
		either $x\in S$ or the degree-one neighbor of $x$ is in $S$.
		Also, note that in an edge gadget, we have one vertex in $S$
		for each pair of adjacent degree-two vertices, and $a\in S$ or
		$b\in S$. Therefore, an edge gadget of an edge of weight $\gamma$ 
		must contain at least $2\gamma+1$ vertices in $S$. 
		It follows that we cannot place both a vertex $x\in V$ and
		its degree-one neighbor in $S$, and that we use exactly
		$2w(e)+1$ vertices from an edge gadget of $e$.
		
		Now, if we would use a degree-one neighbor of $x\in V$, we can
		swap this vertex with $x$, and also have a solution. Thus $V\subseteq S$, and each vertex $x\in V\cap S$ has its incident degree one edge associated to it.
		Consider an edge $e=\{x,y\}$. If $b\in S$, then orient the edge as $xy$; otherwise, $a\in S$ and we orient the edge as $yx$.
		If $b\in S$, $a\not\in S$, and thus all neighbors of $a$ are in $S$. This means that the edges from $x$ to the gadget must be
		covered by $x$, and thus $x$ have $w(e)$ edges from the gadget
		assigned to it. The number of gadget edges assigned to $x$ thus
		equals its weighted outdegree, which is at most $t(x)$. The lemma now follows.    
	\end{proof}
	
	As the transformation maintains planarity, and increases the
	outerplanarity of a graph by at most 1, the following result
	is obtained.
	
	\begin{theorem}
		\textsc{Capacitated Vertex Cover} is
		XALP-complete for planar graphs with outerplanarity as parameter.
	\end{theorem}
	
	\subsection{$f$-Domination and $k$-Domination}
	\label{ssec:fkdomination}
	In contrast to capacitated versions of dominating set where vertices can dominate only a limited number of neighbours, in the \textsc{$f$-Dominating Set} problem vertices not in the solution must be dominated multiple times. 
	
	\defproblem{$f$-\textsc{Dominating Set}}{Graph $G=(V,E)$, demand function $f:V\rightarrow \mathbb{N}$, integer $k$}{Is there a set $S\subseteq V$ such that $|S|\geq k$ and for each $v\in V$, we have $v\in S$ or $|N(v)\cap S|\leq f(v)$?}
	
	A special case of $f$-\textsc{Dominating Set} is the $k$-\textsc{Dominating Set} problem, where we $f(v)$ is equal for all vertices $v\in V$. Note that if $k$ is fixed, standard (dynamic programming) techniques
	give an FPT algorithm with treewidth as parameter.
	
	\defproblem{$k$-\textsc{Dominating Set}}{Graph $G=(V, E)$, integers $\ell, k$}{Is there a set $S\subseteq V$ such that $|S|\leq \ell$ and for each $v\in V$ we have $v\in S$ or $|N(v)\cap S|\geq k$.}
	
	The notion of a $k$-dominating set was introduced in 1985 by Fink and Jacobson~\cite{FinkJ85} and was mostly studied from a graph
	theoretic perspective. If $k$ is fixed, then 
	the \textsc{$k$-Dominating Set} problem can be solved in
	$O(k^{O(tw)}n)$ time, using dynamic programming on tree decompositions~\cite{Telle94}.
	The generalization of
	$f$-domination was introduced by Zhou, see e.g.~\cite{ChenZ98}.
	Again, the \textsc{$f$-Dominating Set} can be easily seen to
	be fixed parameter tractable when we take a combined parameter
	of treewidth and the maximum value of $f$, $k=\max_{v\in V} f(v)$,
	as parameter, also by using dynamic programming
	on tree decompositions.
	
	The hardness proof is quite similar.
	An edge of weight $\gamma$ is replaced by the following structure.
	We take a triangle with three vertices, say $a$, $b$, and $c$. Take $\gamma$ copies of a path
	of three vertices, and make the first vertices on these paths adjacent to $x$, and the last (third)
	vertices of these paths adjacent to $a$. Also, take $\gamma$ copies of a path
	of three vertices, and make the first vertices on these paths adjacent to $b$, and the last (third)
	vertices of these paths adjacent to $y$. See Figure~\ref{fig:f-domination} (left side) for an example.
	For all vertices in the new edge gadgets, we set their $f$-values equal to 1. For an original
	vertex $v\in V$, we set $f(v)=t(v)$.
	
	Let $H$ be the resulting graph, with domination demand function $f$.
	
	\begin{figure}[htb]
		\centering
		\includegraphics{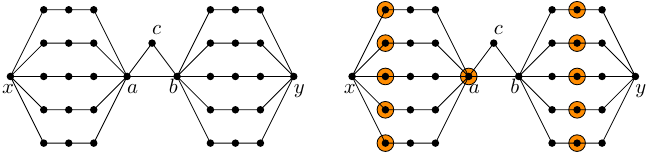}
		\caption{Example transformation of an edge for the \textsc{$f$-Dominating Set} problem. The right side corresponds to an orientation of
			the edge from $x$ to $y$. All new vertices have $f(v)=1$;
			$f(x)=t(x)$; $f(y)=t(y)$}
		\label{fig:f-domination}
	\end{figure}
	
	\begin{lemma}\label{lem:f-dom}
		$H$ has an $f$-dominating set of size $\sum_{e\in E}(2w(e)+1)$, if and only if $G$ has
		an orientation with each vertex outdegree at least $t(v)$.
	\end{lemma}
	
	\begin{proof}
		Consider an orientation of $G$ such that each vertex has outdegree at least $t(v)$. We will construct an $f$-dominating set $S$ in $H$ as follows. For each edge $xy\in E(G)$ that is oriented from $x$ to $y$, we add to $S$ the first vertex of each path from $x$ to $a$, the vertex $a$, and the middle vertex of each path from $b$ to $y$. Note that for each such edge we added $2w(e)+1$ vertices to $S$, so $|S|=\sum_{e\in E(G)}(2w(e)+1)$. It is easy to see that all the new vertices in $H$ have a neighbour in $S$. Consider a vertex $x\in V(G)$ and an edge $xy\in E(G)$ oriented from $x$ to $y$. This edge contributes $w(e)$ to the outdegree of $x$ in $G$, and it contributes $w(e)$ to the number of neighbours in the set $S$ in $H$. If an edge $xy\in E(G)$ is oriented from $y$ to $x$, it contributes neither to the outdegree of $x$ in $G$ nor to the number of neighbours of $v$ in $S$ in $H$. The other direction (i.e. constructing an orientation of $G$ from an $f$-dominating set in $H$) is analogous.   
	\end{proof}
	
	\begin{theorem}
		The \textsc{$f$-Dominating Set} problem is:
		\begin{enumerate}
			\item XALP-complete for planar graphs with outerplanarity as parameter.
			\item XNLP-complete with pathwidth as parameter.
			\item XALP-complete with treewidth as parameter.
		\end{enumerate}
	\end{theorem}
	\begin{proof}
		We reduce from \textsc{Target Outdegree Orientation}. We construct the graph $H$ as in the above description. Using Lemma~\ref{lem:f-dom}, we obtain the desired result.
	\end{proof}
	
	\begin{corollary}
		The \textsc{$k$-Dominating Set} problem is:
		\begin{enumerate}
			\item XNLP-hard for planar graphs with outerplanarity as parameter.
			\item XNLP-complete with pathwidth as parameter.
			\item XALP-complete with treewidth as parameter.
		\end{enumerate}
	\end{corollary}
	
	\begin{proof}
		The reduction for the hardness proofs is the following.
		
		Suppose we are given an instance of \textsc{$f$-Dominating Set}, with $G=(V,E)$ a graph,
		$f$ a demand function, and $\ell$ the maximum solution size.
		Let $k=\max_{v\in V} f(v)$. We may
		assume $k>1$, otherwise the problem is fixed parameter tractable (this can be shown by standard techniques).
		
		Now, to each vertex $v\in V$, add $k-f(v)$ new vertices, only adjacent to $v$. Increase
		$\ell$ by the total number of added new vertices. Correctness follows by observing that
		all new vertices must be an element of the dominating set (they have only one neighbour in $G$), and now each vertex in $V$ must be in the dominating set, or has at least
		$f(v)$ neighbours in $V$ in the dominating set (as it has $k-f(v)$ neighbours that are new vertices in the dominating set.)
	\end{proof}
	
	\subsection{Target Set Selection}
	\label{ssec:tss}
	The \textsc{Target Set Selection} problem models the viral marketing process, i.e. the style of promotion relying on consumers recommending the product to their social network. Informally, the problem setup is as follows. We are given a graph corresponding to a social network, and we want to advertise our product by giving it away to at most $k$ ``influencers'' in this network. They will spread the word about our product and convince others to buy it --- each person will buy the product if it was recommended by at least a certain number (i.e. threshold) of their friends. 
	
	Formally, given a graph $G=(V, E)$, a threshold function $t:V\rightarrow \mathbb{N}$ and $S\subseteq V$, the \emph{activation process in $G$ starting with $S$} is a sequence of subsets $Active[0]\subseteq Active[1]\subseteq\dots$ such that $Active[0]=S$ and a vertex $v$ belongs to $Active[i]$ if it belongs to $Active[i-1]$ or it has at least $t(v)$ neighbours in $Active[i-1]$. We repeat this process until we get $Active[j]=V$ or $Active[j]=Active[j-1]$ for some $j\in\mathbb{N}$, and we define $Active(S)=Active[j]$.
	
	\defproblem{\textsc{Target Set Selection}}{Graph $G=(V, E)$, a threshold $t:V\rightarrow \mathbb{N}$, integer $k$}{Is there a set $S\subseteq V$ such that $Active(S)=V$?}
	
	Ben-Zwi et al.~\cite{Ben-ZwiHLN11} gave an XP algorithm with treewidth as parameter, and an
	$n^{\Omega (\sqrt{tw})}$ lower bound. Their proof also implies W[1]-hardness.
	
	For our hardness proofs, we use an edge gadget that is again similar
	to previous edge gadgets; it is illustrated in Figure~\ref{fig:tss}.
	Here, we reduce from \textsc{Minimum Outdegree Orientation}.
	Original vertices have a threshold which is equal to their target
	outdegree; the vertices $a$ and $b$ have a threshold equal to their
	degree, i.e., the weight of the edge plus one, and the vertices of degree two in the gadget have a threshold of 1. We will call the obtained graph $H$.
	
	\begin{figure}
		\centering
		\includegraphics{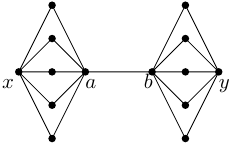}
		\caption{Transformation for Target Set Selection. The example transforms an edge of weight $\gamma=5$.
			$a$ and $b$ have threshold $\gamma+1=6$; $x$ has a threshold equal to $t(x)$; $y$ has a threshold 
			equal to $t(y)$; the other (degree 2) vertices have threshold 1.}
		\label{fig:tss}
	\end{figure}
	
	\begin{lemma}
		We can activate $H$ by initially activating $|E(G)|$ 
		vertices, if and only if $G$ has an orientation with each
		vertex $v\in V$ total weighted outdegree at least $t(v)$.
	\end{lemma}
	
	\begin{proof}
		Suppose we can activate $H$ by initially activating
		a set $S$ such that $|S|=|E(G)|$.
		
		Firstly note that in each edge gadget, we must initially 
		activate either $a$ or $b$: otherwise, at least one of them cannot
		reach the threshold. Now, in each edge gadget, exactly one of the vertices $a$ and $b$ is in $S$, as otherwise we would have too many vertices in $S$. 
		Consider an edge gadget where $b\in S$. In order to activate $a$, we must activate all of its degree 2 neighbours, which means we must activate $x$ before that. In other words, in order to activate $x$, we cannot use any of its neighbours in this gadget. Thus, to activate $x$, we must use its degree 2 neighbours from edge gadgets where $a\in S$ (note that in these gadgets all the degree 2 vertices between $a$ and $x$ are activated in the same round). 
		
		This gives us an orientation of edges in $G$: for each edge $xy\in E(G)$, orient the edge from $x$ to $y$ if the corresponding vertex $a$ is in $S$ and from $y$ to $x$ otherwise. By the above arguments, there are at least $t(x)$ neighbours of $x$ activated before $x$, which means that the sum of weights of all outgoing edges from $x$ is $t(x)$. The reverse direction is analogous. 
	\end{proof}
	
	\begin{theorem}
		The \textsc{Target Set Selection} problem is:
		\begin{enumerate}
			\item XALP-complete for planar graphs with outerplanarity as parameter.
			\item XNLP-complete with pathwidth as parameter.
			\item XALP-complete with treewidth as parameter.
		\end{enumerate}
	\end{theorem}
	
	\begin{proof}
		As the transformation described above can be done in
		logarithmic space, hardness follows from the hardness for
		\textsc{Minimum Outdegree Orientation}.
		
		Membership follows from a modification of the XP algorithm
		in \cite{Ben-ZwiHLN11}. Instead of building the entire DP
		table for a path decomposition, we guess the `next element'. 
		In a tree decomposition, we traverse the tree in post-order,
		guessing elements instead of building the entire DP table;
		in a (join) node with two children, after handling the left
		branch, we store its result on a stack, then handle the right
		branch, and then combine the results of left and right branch.
	\end{proof}

\section{Conclusion}
	\label{sec:conclusion}
	In this paper, we showed XALP-completeness for several problems on planar graphs parameterized by outerplanarity.
	In a number of cases, these
	were problems not yet established to be XALP-complete for treewidth, 
	and thus for these, we also obtained new results for the
	general class of graphs. 
	
	The results are also interesting in the light of a conjecture by Pilipczuk and 
	Wrochna~\cite{PilipczukW18}, which claims that an XNLP-hard problem has no XP algorithm 
	that uses $O(f(k)n^{O(1)})$ space.
	Our results are negative in the sense that it appears that for the problems studied in this setting,
	going from graphs of bounded treewidth to graphs of bounded outerplanarity, the complexity does not drop.
	
	We expect that more problems that are in XP with outerplanarity as parameter are complete for XALP,
	and leave finding more examples as an interesting open problem. One candidate is \textsc{Multicommodity Flow}:
	in \cite{BodlaenderMOPL23}, it was shown that it is XALP-complete for treewidth
	and XNLP-complete for pathwidth; capacities are given in unary. The
	result already holds for 2 commodities. The gadgets used in that
	proof are non-planar, so establishing the complexity of \textsc{Multicommodity Flow} with outerplanarity seems to need new techniques.
	
	Another interesting direction for further research is the following. Problems with an algorithm 
	with running time of the form $O(2^{O(\ell)} n^{O(1)})$ with $\ell$ the treewidth or outerplanarity
	are in XP when we take `logarithmic treewidth' or `logarithmic outerplanarity' as parameter, i.e.,
	the parameter is $\frac{\ell}{\log n}$. In \cite{BodlaenderGNS22a}, it was shown that 
	\textsc{Dominating Set} and \textsc{Independent Set} are XNLP-complete with logarithmic pathwidth
	as parameter; this translates to XALP-completeness for these problems with logarithmic treewidth
	as parameter~\cite{BodlaenderGJPP22a}. As these proofs produce non-planar graphs, it is open
	whether \textsc{Dominating Set} and \textsc{Independent Set} are XALP-complete when we take
	logarithmic outerplanarity as parameter.

    Finally, Jansen et al.~\cite{JansenKKLMS23}, show that the \textsc{Upward Planarity Testing} problem
    and the \textsc{Orthogonal Planarity Testing} problem are W[1]-hard Parameterized by Treewidth. 
    These results are proven by a chain of reductions, via \textsc{All-or-Nothing Flow}, and \textsc{Circulating Orientation}, 
    for special instances. Combining the reductions in our paper and those from Jansen et al.~\cite{JansenKKLMS23} seems
    to lead to XALP-completeness proofs for these two problems with treewidth as parameter, under parameterized reductions.
    Whether these results would lead to XALP-completeness with treewidth as parameter under pl-reductions, and/or
    to XALP-completeness with outerplanarity as parameter (for the upward and orthogonal planarity testing) is left for
    further exploration.

	\bibliographystyle{abbrvurl}
	\bibliography{papers}
\end{document}

%% file: macros.tex
\usepackage{enumitem}
\usepackage{wrapfig}
\usepackage{todonotes}
\usepackage{soul}
\usepackage{makecell}
\usepackage{todonotes}
\usepackage{array}
\newcolumntype{L}[1]{>{\raggedright\arraybackslash}m{#1}}
\newcolumntype{C}[1]{>{\centering\arraybackslash}m{#1}}

\usepackage{amsmath, amssymb, graphics, graphicx}
\usepackage{color}
\usepackage{comment}
\usepackage{cases}
\usepackage{xparse}
\usepackage{xargs}
\usepackage{appendix}
\usepackage{xcolor}

\usepackage{relsize}

\usepackage{tikz}
\usetikzlibrary{positioning, calc, shapes, fit}

\usepackage{standalone}
\usepackage{import}
\usetikzlibrary{positioning, shapes.geometric, shapes, fit, arrows, decorations.markings, arrows.meta}

\newcommand{\defproblem}[3]{
	\vspace{2mm}
	\vspace{1mm}
	\noindent\fbox{
		\begin{minipage}{0.95\textwidth}
			#1 \\
			{\bf{Input:}} #2  \\
			{\bf{Task:}} #3
		\end{minipage}
	}
	\vspace{2mm}
}

\usepackage{amsthm}
\theoremstyle{remark}
\newtheorem*{claimproof}{Proof of Claim}


%% file: figures/grid.tex
\begin{tikzpicture}
	\tikzset{
		vx/.style={draw, circle, inner sep=0pt, outer sep=0pt, fill, minimum size=2pt}
	}
	\def\u{0.7}
	\foreach \i in {1,3,5}{
		  \draw (\i*\u, \u)--(\i*\u, 7*\u);}
        \foreach \i in {1,3,5,7}{
		  \draw (0, \i*\u)--(7*\u, \i*\u);}
	\foreach \i in {2,4,6}{
		  \draw[color=green] (\i*\u, \u)--(\i*\u, 7*\u);
		  \draw[color=green] (0, \i*\u)--(7*\u, \i*\u);}
	\node(lv) at (\u, 3*\u) [label=below left:$\ell_v$, vx] {};
	\node(rv) at (7*\u, 3*\u)  {};
	\draw[line width=2pt] (lv)--(rv);
	\node(tv) at (\u, 7*\u){};
	\node(bv) at (\u, \u){};
	\draw[line width=2pt](tv)--(bv);
	\node (lu) at (3*\u, 5*\u)[vx, label=below left:$\ell_u$]{};
	\node (ru) at (7*\u, 5*\u)[vx]{};
	\draw[line width=2pt](lu)--(ru);
	\node(tu) at (3*\u, 7*\u){};
	\node(bu) at (3*\u, \u){};
	\draw[line width=2pt](bu)--(tu);
	\node(tuv) at (3*\u, 3*\u)[vx, label=below left:$b_2$]{};
	\node[right=\u of lv, label=below left:$g_1$](b) {};
	\node[right=\u of tuv, label=below left:$g_2$](b1){};
	\node at (\u, \u) [label=below: $J_v$]{};
	\node at (3*\u, \u) [label=below: $J_u$]{};
	\node at (7*\u, 3*\u) [label=below left: $I_v$]{};
	\node at (7*\u, 5*\u) [label=below left: $I_u$]{};
	\end{tikzpicture}

%% file: figures/scattered2.tex
\begin{tikzpicture}
	\tikzset{
		vx/.style={draw, circle, inner sep=0pt, outer sep=0pt, fill, minimum size=2pt}
	}
	\def\u{1}
        \def \a{1.5}
        \def \b{3}
	\draw (0,0) ellipse (1.5 and 3);
	\node[vx] at ($(0,0)+(45:1.5 and 3)$)(a11){};
	\node[vx] at ($(0,0)+(15:1.5 and 3)$)(a12){};
	\node[vx] at ($(0,0)+(-15:1.5 and 3)$)(a21){};
	\node[vx] at ($(0,0)+(-45:1.5 and 3)$)(a22){};
	
	\node[vx, label=left:$C_u(v)_{1,1}$](a11) at ($(0,0)+(45:1.5 and 3)$){};	
	\node[vx, label=left:$C_u(v)_{1,2}$](a12) at ($(0,0)+(15:1.5 and 3)$){};
	\node[vx, label=left:$C_u(v)_{2,1}$](a21) at ($(0,0)+(-15:1.5 and 3)$){};
	\node[vx, label=left:$C_u(v)_{2,2}$](a22) at ($(0,0)+(-45:1.5 and 3)$){};
	\draw (5,0) ellipse (1.5 and 3);
	\node[vx, label=right:$C_v(u)_{4,2}$](b11) at ($(5,0)+(135:1.5 and 3)$){};	
	\node[vx, label=right:$C_v(u)_{4,1}$](b12) at ($(5,0)+(165:1.5 and 3)$){};
	\node[vx, label=right:$C_v(u)_{3,2}$](b21) at ($(5,0)+(195:1.5 and 3)$){};
	\node[vx, label=right:$C_v(u)_{3,1}$](b22) at ($(5,0)+(225:1.5 and 3)$){};

	\draw (a11) edge node foreach \i in {1,...,6}[pos=0.14*\i, vx](c\i){} (b12);
	\draw (a21) edge node foreach \i in {1,...,6}[pos=0.14*\i, vx](d\i){}(b21);
	\foreach \i in {1,2,5,6}
	{
		\draw (c\i)-- ++(0, \u) node foreach \i in {0.3, 0.6, 1}[pos=\i, vx]{};	
		\draw (d\i)-- ++(0, -\u) node foreach \i in {0.3, 0.6, 1}[pos=\i, vx]{};
	}
	\node[vx, label=above left:$C_u(w)_{1,2}$](e11) at ($(0,0)+(135:1.5 and 3)$){};
	\node[vx, label=above left:$C_u(w)_{1,1}$](e12) at ($(0,0)+(165:1.5 and 3)$){};
\end{tikzpicture}

%% file: figures/reductions.tex
\begin{tikzpicture}[every text node part/.style={text centered}]
\def\u{1}
\tikzset
{
	->-/.style={decoration={markings,mark=at position 0.5 with {\arrow{Straight Barb}}},
		postaction={decorate}}
}
\node[draw](v1) at (0,0){\textsc{All-or-Nothing Flow}};
\node[draw, below=\u of v1, text width=3.5cm](v2){\textsc{All-or-Nothing Flow} cap. $\leq 2$};
\draw[->-] (v1)--(v2);
\node[draw, left=\u of v2, text width=3.5cm](v3){\textsc{Max. Outdegree Orientation}};
\draw[->-](v2)--(v3);
\node[draw, right=\u of v2, text width=4cm](v4){\textsc{Target Outdegree Orientation}};
\node[draw, below=\u of v4, text width=4cm](vco){\textsc{Circulating Orientation}};
\draw[->-](v4)--(vco);
\draw[->-](v2)--(v4);
\node[draw, below right=3*\u and \u of v2, text width=3.5cm](v5){\textsc{Min. Outdegree Orientation}};
\draw[->-] (v2)--(v5);
\node[draw, below=\u of v3, text width=3.5cm](v6){\textsc{Cap. (Red-Blue) Dominating Set}};
\draw[->-](v3)--(v6);
\node[draw, below right=1.3*\u of v3](v7){\textsc{Cap. Vertex Cover}};
\draw[->-] (v3)--(v7);
\node[draw, below=\u of v5](v8){$f$-\textsc{Dominating Set}};
\draw[->-] (v5)--(v8);
\node[draw, below=\u of v8](v9){$k$-\textsc{Dominating Set}};
\draw[->-](v8)--(v9);
\node[draw, below left=0.3*\u of v5](v10){\textsc{Target Set Selection}};
\draw[->-](v5)--(v10);
\end{tikzpicture}

%% file: main.bbl
\begin{thebibliography}{10}

\bibitem{AlberFN04}
J.~Alber, M.~R. Fellows, and R.~Niedermeier.
\newblock Polynomial-time data reduction for dominating set.
\newblock {\em Journal of the {ACM}}, 51(3):363--384, 2004.
\newblock \href {https://doi.org/10.1145/990308.990309} {\path{doi:10.1145/990308.990309}}.

\bibitem{alexandersson2020np}
P.~Alexandersson.
\newblock {NP}-complete variants of some classical graph problems.
\newblock {\em arXiv}, 2020.
\newblock \href {https://arxiv.org/abs/2001.04120} {\path{arXiv:2001.04120}}.

\bibitem{Baker94}
B.~S. Baker.
\newblock Approximation algorithms for {NP}-complete problems on planar graphs.
\newblock {\em Journal of the {ACM}}, 41(1):153--180, 1994.
\newblock \href {https://doi.org/10.1145/174644.174650} {\path{doi:10.1145/174644.174650}}.

\bibitem{Ben-ZwiHLN11}
O.~Ben{-}Zwi, D.~Hermelin, D.~Lokshtanov, and I.~Newman.
\newblock Treewidth governs the complexity of target set selection.
\newblock {\em Discrete Optimization}, 8(1):87--96, 2011.
\newblock \href {https://doi.org/10.1016/J.DISOPT.2010.09.007} {\path{doi:10.1016/J.DISOPT.2010.09.007}}.

\bibitem{bodlaender1998partial}
H.~L. Bodlaender.
\newblock A partial $k$-arboretum of graphs with bounded treewidth.
\newblock {\em Theoretical Computer Science}, 209(1-2):1--45, 1998.
\newblock \href {https://doi.org/10.1016/S0304-3975(97)00228-4} {\path{doi:10.1016/S0304-3975(97)00228-4}}.

\bibitem{BodlaenderCW22a}
H.~L. Bodlaender, G.~Cornelissen, and M.~van~der Wegen.
\newblock Problems hard for treewidth but easy for stable gonality.
\newblock In M.~A. Bekos and M.~Kaufmann, editors, {\em Proceedings 48th International Workshop on Graph-Theoretic Concepts in Computer Science, {WG} 2022}, volume 13453 of {\em Lecture Notes in Computer Science}, pages 84--97. Springer, 2022.
\newblock \href {https://doi.org/10.1007/978-3-031-15914-5\_7} {\path{doi:10.1007/978-3-031-15914-5\_7}}.

\bibitem{BodlaenderGJ22listcoloring}
H.~L. Bodlaender, C.~Groenland, and H.~Jacob.
\newblock List colouring trees in logarithmic space.
\newblock In S.~Chechik, G.~Navarro, E.~Rotenberg, and G.~Herman, editors, {\em Proceedings 30th Annual European Symposium on Algorithms, {ESA} 2022}, volume 244 of {\em LIPIcs}, pages 24:1--24:15. Schloss Dagstuhl - Leibniz-Zentrum f{\"{u}}r Informatik, 2022.
\newblock \href {https://doi.org/10.4230/LIPICS.ESA.2022.24} {\path{doi:10.4230/LIPICS.ESA.2022.24}}.

\bibitem{BodlaenderGJJL22}
H.~L. Bodlaender, C.~Groenland, H.~Jacob, L.~Jaffke, and P.~T. Lima.
\newblock {XNLP}-completeness for parameterized problems on graphs with a linear structure.
\newblock In H.~Dell and J.~Nederlof, editors, {\em Proceedings 17th International Symposium on Parameterized and Exact Computation, {IPEC} 2022}, volume 249 of {\em LIPIcs}, pages 8:1--8:18. Schloss Dagstuhl - Leibniz-Zentrum f{\"{u}}r Informatik, 2022.
\newblock \href {https://doi.org/10.4230/LIPIcs.IPEC.2022.8} {\path{doi:10.4230/LIPIcs.IPEC.2022.8}}.

\bibitem{BodlaenderGJPP22a}
H.~L. Bodlaender, C.~Groenland, H.~Jacob, M.~Pilipczuk, and M.~Pilipczuk.
\newblock On the complexity of problems on tree-structured graphs.
\newblock In H.~Dell and J.~Nederlof, editors, {\em Proceedings 17th International Symposium on Parameterized and Exact Computation, {IPEC} 2022}, volume 249 of {\em LIPIcs}, pages 6:1--6:17. Schloss Dagstuhl - Leibniz-Zentrum f{\"{u}}r Informatik, 2022.
\newblock \href {https://doi.org/10.4230/LIPIcs.IPEC.2022.6} {\path{doi:10.4230/LIPIcs.IPEC.2022.6}}.

\bibitem{BodlaenderGNS22a}
H.~L. Bodlaender, C.~Groenland, J.~Nederlof, and C.~M.~F. Swennenhuis.
\newblock Parameterized problems complete for nondeterministic {FPT} time and logarithmic space.
\newblock In {\em Proceedings 62nd {IEEE} Annual Symposium on Foundations of Computer Science, {FOCS} 2021}, pages 193--204. {IEEE}, 2022.
\newblock \href {https://doi.org/10.1109/FOCS52979.2021.00027} {\path{doi:10.1109/FOCS52979.2021.00027}}.

\bibitem{BodlaenderLP09}
H.~L. Bodlaender, D.~Lokshtanov, and E.~Penninkx.
\newblock Planar capacitated dominating set is \emph{W}[1]-hard.
\newblock In J.~Chen and F.~V. Fomin, editors, {\em Proceedings 4th International Workshop on Parameterized and Exact Computation, {IWPEC} 2009}, volume 5917 of {\em Lecture Notes in Computer Science}, pages 50--60. Springer, 2009.
\newblock \href {https://doi.org/10.1007/978-3-642-11269-0\_4} {\path{doi:10.1007/978-3-642-11269-0\_4}}.

\bibitem{BodlaenderMOPL23}
H.~L. Bodlaender, I.~Mannens, J.~J. Oostveen, S.~Pandey, and E.~J. van Leeuwen.
\newblock The parameterised complexity of integer multicommodity flow.
\newblock In N.~Misra and M.~Wahlstr{\"{o}}m, editors, {\em Proceedings 18th International Symposium on Parameterized and Exact Computation, {IPEC} 2023}, volume 285 of {\em LIPIcs}, pages 6:1--6:19. Schloss Dagstuhl - Leibniz-Zentrum f{\"{u}}r Informatik, 2023.
\newblock \href {https://doi.org/10.4230/LIPICS.IPEC.2023.6} {\path{doi:10.4230/LIPICS.IPEC.2023.6}}.

\bibitem{ChenZ98}
B.~Chen and S.~Zhou.
\newblock Upper bounds for $f$-domination number of graphs.
\newblock {\em Discrete Mathematics}, 185(1-3):239--243, 1998.
\newblock \href {https://doi.org/10.1016/S0012-365X(97)00204-5} {\path{doi:10.1016/S0012-365X(97)00204-5}}.

\bibitem{didimo2019hv}
W.~Didimo, G.~Liotta, and M.~Patrignani.
\newblock {HV}-planarity: Algorithms and complexity.
\newblock {\em Journal of Computer and System Sciences}, 99:72--90, 2019.
\newblock \href {https://doi.org/10.1016/j.jcss.2018.08.003} {\path{doi:10.1016/j.jcss.2018.08.003}}.

\bibitem{DomLSV08}
M.~Dom, D.~Lokshtanov, S.~Saurabh, and Y.~Villanger.
\newblock Capacitated domination and covering: {A} parameterized perspective.
\newblock In M.~Grohe and R.~Niedermeier, editors, {\em Proceedings 3rd International Workshop on Parameterized and Exact Computation, Third International Workshop, {IWPEC} 2008}, volume 5018 of {\em Lecture Notes in Computer Science}, pages 78--90. Springer, 2008.
\newblock \href {https://doi.org/10.1007/978-3-540-79723-4\_9} {\path{doi:10.1007/978-3-540-79723-4\_9}}.

\bibitem{downey1995fixed}
R.~G. Downey and M.~R. Fellows.
\newblock Fixed-parameter tractability and completeness {I}: Basic results.
\newblock {\em SIAM Journal on computing}, 24(4):873--921, 1995.
\newblock \href {https://doi.org/10.1137/S0097539792228228} {\path{doi:10.1137/S0097539792228228}}.

\bibitem{DowneyF95fixed2}
R.~G. Downey and M.~R. Fellows.
\newblock Fixed-parameter tractability and completeness {II:} on completeness for {W[1]}.
\newblock {\em Theoretical Computer Science}, 141(1{\&}2):109--131, 1995.
\newblock \href {https://doi.org/10.1016/0304-3975(94)00097-3} {\path{doi:10.1016/0304-3975(94)00097-3}}.

\bibitem{ElberfeldST15}
M.~Elberfeld, C.~Stockhusen, and T.~Tantau.
\newblock On the space and circuit complexity of parameterized problems: Classes and completeness.
\newblock {\em Algorithmica}, 71(3):661--701, 2015.
\newblock \href {https://doi.org/10.1007/s00453-014-9944-y} {\path{doi:10.1007/s00453-014-9944-y}}.

\bibitem{ErdosT41}
P.~Erd\"{o}s and P.~Tur\'{a}n.
\newblock On a problem of {S}idon in additive number theory, and on some related problems.
\newblock {\em Journal of the London Mathematical Society}, s1-16(4):212--215, 1941.
\newblock \href {https://doi.org/10.1112/jlms/s1-16.4.212} {\path{doi:10.1112/jlms/s1-16.4.212}}.

\bibitem{eto2014distance}
H.~Eto, F.~Guo, and E.~Miyano.
\newblock Distance-$d$ independent set problems for bipartite and chordal graphs.
\newblock {\em Journal of Combinatorial Optimization}, 27(1):88--99, 2014.
\newblock \href {https://doi.org/10.1007/S10878-012-9594-4} {\path{doi:10.1007/S10878-012-9594-4}}.

\bibitem{FinkJ85}
J.~F. Fink and M.~S. Jacobson.
\newblock $n$-{D}omination in graphs.
\newblock In {\em Graph Theory with Application to Algorithms and Computer Science}, pages 282--300. John Wiley and Sons, 1985.

\bibitem{FominGLS14}
F.~V. Fomin, P.~A. Golovach, D.~Lokshtanov, and S.~Saurabh.
\newblock Almost optimal lower bounds for problems parameterized by clique-width.
\newblock {\em {SIAM} Journal on Computing}, 43(5):1541--1563, 2014.
\newblock \href {https://doi.org/10.1137/130910932} {\path{doi:10.1137/130910932}}.

\bibitem{JansenKKLMS23}
B.~M.~P. Jansen, L.~Khazaliya, P.~Kindermann, G.~Liotta, F.~Montecchiani, and K.~Simonov.
\newblock Upward and orthogonal planarity are {W[1]}-hard parameterized by treewidth.
\newblock In M.~A. Bekos and M.~Chimani, editors, {\em 31st International Symposium on Graph Drawing and Network Visualization, {GD} 2023, Proceedings, Part {II}}, volume 14466 of {\em Lecture Notes in Computer Science}, pages 203--217. Springer, 2023.
\newblock \href {https://doi.org/10.1007/978-3-031-49275-4\_14} {\path{doi:10.1007/978-3-031-49275-4\_14}}.

\bibitem{JansenS97}
K.~Jansen and P.~Scheffler.
\newblock Generalized coloring for tree-like graphs.
\newblock {\em Discrete Applied Mathematics}, 75(2):135--155, 1997.
\newblock \href {https://doi.org/10.1016/S0166-218X(96)00085-6} {\path{doi:10.1016/S0166-218X(96)00085-6}}.

\bibitem{Kammer07}
F.~Kammer.
\newblock Determining the smallest $k$ such that \emph{G} is $k$-outerplanar.
\newblock In L.~Arge, M.~Hoffmann, and E.~Welzl, editors, {\em Proceedings 15th Annual European Symposium on Algorithms, {ESA} 2007}, volume 4698 of {\em Lecture Notes in Computer Science}, pages 359--370. Springer, 2007.
\newblock \href {https://doi.org/10.1007/978-3-540-75520-3\_33} {\path{doi:10.1007/978-3-540-75520-3\_33}}.

\bibitem{katsikarelis2022structurally}
I.~Katsikarelis, M.~Lampis, and V.~T. Paschos.
\newblock Structurally parameterized $d$-scattered set.
\newblock {\em Discrete Applied Mathematics}, 308:168--186, 2022.
\newblock \href {https://doi.org/10.1016/J.DAM.2020.03.052} {\path{doi:10.1016/J.DAM.2020.03.052}}.

\bibitem{Kloks94}
T.~Kloks.
\newblock {\em Treewidth, Computations and Approximations}, volume 842 of {\em Lecture Notes in Computer Science}.
\newblock Springer, 1994.
\newblock \href {https://doi.org/10.1007/BFB0045375} {\path{doi:10.1007/BFB0045375}}.

\bibitem{PilipczukW18}
M.~Pilipczuk and M.~Wrochna.
\newblock On space efficiency of algorithms working on structural decompositions of graphs.
\newblock {\em {ACM} Transactions on Computation Theory}, 9(4):18:1--18:36, 2018.
\newblock \href {https://doi.org/10.1145/3154856} {\path{doi:10.1145/3154856}}.

\bibitem{Telle94}
J.~A. Telle.
\newblock Complexity of domination-type problems in graphs.
\newblock {\em Nordic Journal of Computing}, 1(1):157--171, 1994.
\newblock URL: \url{https://www.cs.helsinki.fi/njc/njc1.html}.

\end{thebibliography}
